\documentclass[11pt]{article}
\usepackage{amsmath,amsfonts,amssymb,amsthm}
\usepackage{bbm}
\usepackage{algorithm}
\usepackage{algorithmic}

\usepackage{hyperref}
\hypersetup{colorlinks,citecolor=blue}
\usepackage{cleveref}
\usepackage{epstopdf}
\usepackage{graphicx}
\usepackage{color}
\usepackage{amstext}
\usepackage[sans]{dsfont}
\usepackage{mathtools}
\usepackage{subfigure}
\usepackage{url}
\usepackage{dsfont}
\usepackage{float}

\usepackage{colordvi}
\usepackage{bm}

\newtheorem{theorem}{Theorem}[section]
\newtheorem{lemma}{Lemma}[section]
\newtheorem{definition}{Definition}[section]

\newtheorem{observation}{Observation}[section]

\newtheorem{result}{Result}

\newtheorem{innercustomgeneric}{\customgenericname}
\providecommand{\customgenericname}{}
\newcommand{\newcustomtheorem}[2]{%
  \newenvironment{#1}[1]
  {%
   \renewcommand\customgenericname{#2}%
   \renewcommand\theinnercustomgeneric{##1}%
   \innercustomgeneric
  }
  {\endinnercustomgeneric}
}

\newcustomtheorem{customtheorem}{Theorem}
\newcustomtheorem{customlemma}{Lemma}

\newcommand{\set}[1]{\left\{ #1 \right\}}

\newcommand{\eps}{\varepsilon}

\newcommand{\qset}{\mathcal{Q}}
\newcommand{\pset}{\mathcal{P}}

\newcommand{\gset}{\mathcal{G}}

\newcommand{\dset}{\mathcal{D}}
\newcommand{\fset}{\mathcal{F}}

\newcommand{\wel}{\mathsf{Wel}_{\mathsf{T}}}
\newcommand{\welnt}{\mathsf{Wel}_{\mathsf{NT}}}

\newcommand{\opt}{\mathsf{OPT}}
\newcommand{\optlp}{\mathsf{OPT}_{\mathsf{LP}}}

\newcommand{\pathh}{\mathsf{Path}}
\newcommand{\starr}{\mathsf{Star}}
\newcommand{\spider}{\mathsf{Spider}}
\newcommand{\tree}{\mathsf{Tree}}
\newcommand{\cycle}{\mathsf{Cycle}}

\newcommand{\notshow}[1]{}

\newcommand{\optwel}{\mathsf{OPT}}

\newcommand{\ratio}{\textnormal{\textsf{gap}}_{\textnormal{\textsf{T}}}}
\newcommand{\rationt}{\textnormal{\textsf{gap}}_{\textnormal{\textsf{NT}}}}

\newenvironment{prevproof}[2]{\vspace{.1in}\noindent {\em {Proof of {#1}~\ref{#2}:}}}{$\Box$\vskip \belowdisplayskip}

\newenvironment{prog}[1]{
	\begin{minipage}{5.8 in}
		\begin{center}
			{\sc #1}
		\end{center}
	}
	{
	\end{minipage}
}

\advance \oddsidemargin by -0.8in
\newcommand{\myparskip}{4pt}
\parskip \myparskip
\begin{document}


		
		\title{Worst-Case Welfare of Item Pricing in the Tollbooth Problem}
		\author{Zihan Tan\thanks{University of Chicago. Email: {\tt zihantan@uchicago.edu}. Supported in part by NSF grant CCF-1616584.}   \and Yifeng Teng\thanks{Google Research. Email: {\tt yifengt@google.com}} \and Mingfei Zhao\thanks{Google Research. Email: {\tt mingfei@google.com}.}}
		\date{July 7, 2022}
		\maketitle

		\thispagestyle{empty}
		\begin{abstract}
		We study the worst-case welfare of item pricing in the \emph{tollbooth problem}. The problem was first introduced by Guruswami et al.~\cite{guruswami2005profit}, 
		and is a special case of the combinatorial auction in which (i) each of the $m$ items in the auction is an edge of some underlying graph; and (ii) each of the $n$ buyers is single-minded and only interested in buying all edges of a single path. We consider the competitive ratio between the hindsight optimal welfare and the optimal worst-case welfare among all item-pricing mechanisms, when the order of the arriving buyers is adversarial. We assume that buyers own the \emph{tie-breaking} power, i.e. they can choose whether or not to buy the demand path at 0 utility. We prove a tight competitive ratio of $3/2$ when the underlying graph is a single path (also known as the \emph{highway} problem), whereas item-pricing can achieve the hindsight optimal if the seller is allowed to choose a proper tie-breaking rule to maximize the welfare \cite{cheung2008approximation, chawla2017stability}. Moreover, we prove an $O(1)$ upper bound of competitive ratio when the underlying graph is a tree.
		
		For general graphs, we prove an $\Omega(m^{1/8})$ lower bound of the competitive ratio. We show that an $m^{\Omega(1)}$ competitive ratio is unavoidable even if the graph is a grid, or if the capacity of every edge is augmented by a constant factor $c$. The results hold even if the seller has tie-breaking power.


		\end{abstract}


\section{Introduction}



Welfare maximization in combinatorial auctions is one of the central problems in market design. The auctioneer is selling $m$ heterogeneous items to $n$ self-interested buyers. She aims to find a welfare-maximizing allocation, and designs payments to incentivize all buyers to truthfully report their private preferences. We consider the special case that every buyer is \emph{single-minded}, which means that the buyer is only interested in buying a certain subset of items, and has value $0$ if she does not get the whole subset. While the celebrated VCG mechanism~\cite{vickrey1961counterspeculation,clarke1971multipart,groves1973incentives} is truthful and maximizes the welfare, markets in the real world often prefer to implement simpler mechanisms, such as \emph{item pricing}.

In an item-pricing mechanism, each item is given an individual price upfront. Buyers come to the auction sequentially and choose the favorite bundle (among the remaining items) that maximizes their own utilities. For a single-minded buyer, she will purchase her demand set if and only if all items in the set are available, and the total price of the items is at most her value.
The performance of an item-pricing mechanism is measured by its \emph{worst-case welfare}, that is, the minimum sum of buyers' value achieved by the mechanism when the buyers arrive in any adversarial order.
In this setting, the Walrasian equilibrium \cite{arrow1951extension,debreu1951coefficient} for gross-substitutes buyers provides a set of prices as well as a welfare-maximizing allocation, such that every buyer receives her favorite bundle~\cite{kelso1982job}, which allows us to maximize the welfare using an item-pricing mechanism. However, such an equilibrium (or a set of prices) may not exist when the buyers are single-minded. 

In this paper we focus on a special case of the above problem known as the \emph{tollbooth problem}~\cite{guruswami2005profit}, where the set of commodities and demands are represented by edges and paths in a graph.
This setting has found various real-world applications.
For instance, a network service provider wants to sell bandwidth along with the links of a network by pricing on every single link, and each customer is only interested in buying a specific path in the network. In this setting, every item in the auction is an edge in some graph $G$ representing the network, and every buyer is single-minded and is only interested in buying all edges of a specific path of $G$.
We additionally impose the restriction that each commodity (edge) may be given to at most one buyer, which means that if an edge was taken by some previous buyer, then buyers who come afterward whose demand set contains this edge may no longer take it.
Similar to the general case, Walrasian equilibrium is not guaranteed to exist for the tollbooth problem. Moreover, Chen and Rudra~\cite{chen2008walrasian} proved that the problem of determining the existence of a Walrasian equilibrium and the problem of computing such an equilibrium (if it exists) are both NP-hard for the tollbooth problems on general graphs. 
Therefore, an investigation of the power and limits of item pricing is a natural next step towards a deeper understanding of the tollbooth problem. For an item-pricing mechanism, its \emph{competitive ratio} is the ratio between the hindsight optimal welfare, i.e. the optimal welfare achieved by any feasible allocation of items to buyers, and the worst-case welfare of the item-pricing mechanism. In the paper, we study the best (smallest) competitive ratio among all item-pricing mechanisms in a given instance of the tollbooth problem.

In addition to the set of prices, a key factor that can significantly affect the welfare of an item-pricing mechanism is the \emph{tie-breaking} rule. For example, if Walrasian equilibrium exists, item pricing can achieve the optimal welfare, but requires carefully breaking ties among all the favorite bundles of every buyer. However, in real markets, buyers often come to the mechanism themselves and simply purchase an arbitrary favorite bundle that maximizes their own utility. Therefore it is possible that the absence of tie-breaking power may influence the welfare achieved by the mechanism. In the paper, we assume that buyers own the \emph{tie-breaking} power, i.e. they can choose whether or not to buy the demand path at 0 utility. 

\subsection{Our Results and Techniques}
\label{subsec:result}

\paragraph{Competitive Ratio for Path Graphs.} If the seller is allowed to allocate the edges via a proper tie-breaking rule, there indeed exists an item pricing that achieves the offline optimal welfare when the underlying graph $G$ is a single path \cite{cheung2008approximation,chawla2017stability}. Interestingly, we show that this result does not hold when the seller has no tie-breaking power. We present an instance where $G$ is a single path, that no item pricing can achieve more than $2/3$-fraction of the offline optimal welfare (Theorem~\ref{thm: path_no_tie_lower_bound}). On the other hand, we prove that such a $3/2$-approximation is achievable via item pricing (Theorem~\ref{thm:path_no_tie}).

\begin{result}
The competitive ratio for any tollbooth problem instance on a single path is at most $3/2$, if buyers own the tie-breaking power. Moreover, the ratio is tight.
\end{result}

The lower bound is achieved by an example with 3 edges and 4 buyers (\Cref{table:example_path}). The upper bound result is more involved. The proof is enabled by constructing three sets of edge-disjoint paths from all buyers' demand paths such that: (i) every edge in the graph is contained in exactly two paths of the three sets; (ii) each set of paths $\qset$ satisfies a special property called \emph{uniqueness}, which intuitively means that there does not exist another set of paths among the rest paths, whose union is the same as the one of $\qset$. We prove that given any unique set of paths $\qset$, we can design prices to serve all buyers whose demand path is in $\qset$ for any buyers' arrival order. With this lemma, we can design prices that achieve at least 2/3 of the offline optimal, by picking one of the three sets and aiming to serve all buyers whose demand path is in the set.

\paragraph{Competitive Ratio for Trees.} We also study the case when $G$ is a tree. When seller owns the tie-breaking power, we show a tight competitive ratio of $\frac{3}{2}$ (\Cref{thm:tiebreak-tree-ub-lb}). The upper bound is proved by combining \Cref{lem:rounding} with the integrality gap result of multicommodity flow problem on tree~\cite{chekuri2007multicommodity, raghavan1994efficient}. On the other hand, we provide an instance on a star to show the competitive ratio is at least $3/2$. 

When the seller has no tie-breaking power, we prove that the competitive ratio of any tree instance is also upper bounded by an absolute constant. 

\begin{result}
For any $\eps>0$, the competitive ratio for any tollbooth problem instance on a tree is at most $7+\eps$, if buyers own the tie-breaking power.
\end{result}

To prove the result, we start by analyzing the competitive ratio for a special class of graphs called \emph{spider}, which is obtained by replacing each edge in a star graph with a single path. Then given an offline optimal allocation (which corresponds to a set $\pset$ of demand paths) in a tree instance, we partition the paths of $\pset$ into two subsets $\pset=\pset_1\cup\pset_2$, such that for each $t\in \set{1,2}$, the graph obtained by taking the union of all paths in $\pset_t$ is a union of node-disjoint spider graphs. Thus the task of computing the prices on the edges of a tree is reduced to that of a spider, while losing a factor of $2$ in the competitive ratio. 

\paragraph{Competitive Ratio for General Graphs.}
Next we study the tollbooth problem for general graphs. For general single-minded combinatorial auctions, where the demand of every agent is an arbitrary set rather than a single path, the competitive ratio between item pricing and the offline optimal is proved to be $O(\sqrt{m})$ \cite{cheung2008approximation,chin2018approximation}, which is tight up to a constant \cite{feldman2015welfare}. For our problem, we first show that the competitive ratio on general graphs can also be polynomial in the number of its edges, in contrast to the constant competitive ratio for path graphs and trees. This polynomially large ratio is unavoidable even if the graph is a grid and if the seller owns the tie-breaking power. On the other hand, we prove an upper bound of $O(m^{0.4}\log^2 m\log n)$ on the competitive ratio in any tollbooth problem instance (Theorem~\ref{thm:general_graph}). When $n$, the number of buyers in the auction, is subexponential on $m$, our competitive ratio is better than the previous ratio $O(\sqrt{m})$ for general single-minded combinatorial auctions.   

\begin{result}\label{inforthm:general_graph}
There exists a tollbooth problem instance such that the competitive ratio is $\Omega(m^{1/8})$. Moreover, there exist a constant $\alpha\in (0,1)$ and an instance on a grid such that the competitive ratio is $\Omega(m^{\alpha})$. Both results hold even if the seller owns the tie-breaking power. On the other hand, the competitive ratio for any tollbooth problem instance is $O(m^{0.4}\log^2 m\log n)$. Here $m$ is the number of edges and $n$ is the number of buyers.
\end{result}

The hard instance for the $\Omega(m^{1/8})$ competitive ratio
is constructed on a simple series-parallel graph (see Figure~\ref{fig:mahuaintro}). In the hard instance, every buyer demands a path connecting the left-most vertex to the right-most vertex, and the value for each demand path is roughly the same.
The demand paths are constructed carefully, such that (i) each edge of the graph is contained in approximately the same number of demand paths; and (ii) the demand paths are intersecting in some delicate way.
We can then show that, for any price vector $p$, if we denote by $\pset$ the set of affordable (under $p$) demand paths, then either the maximum cardinality of an independent subset of $\pset$ is small, or there is a path in $\pset$ that intersects all other paths in $\pset$. Either way, we can conclude that the optimal worst-case welfare achieved by any item-pricing mechanism is small.
\begin{figure}[h]
	\centering
\includegraphics[scale=0.15]{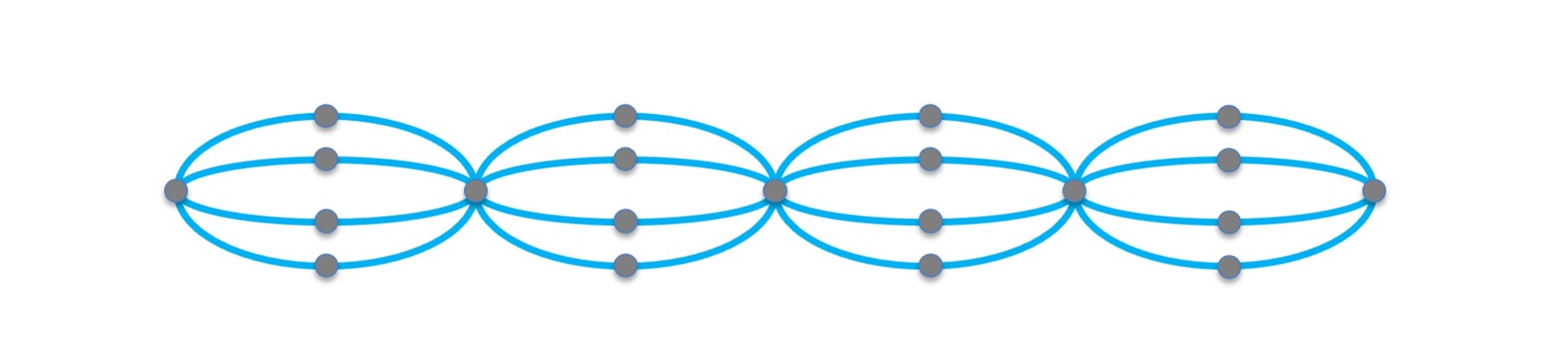}
	\caption{A series-parallel graph with large competitive ratio.\label{fig:mahuaintro}}
\end{figure}

\paragraph{Resource augmentation.} 
At last, we study the problem where the seller has more resource to allocate.
More specifically, comparing to the offline optimal allocation with supply 1 for each item (denoted as $\optwel$), the seller has augmented resources and is allowed to sell $c$ copies of each item to the buyers. In the literature studying offline path allocation problems on graphs (e.g. \cite{chuzhoy2016polylogarithmic}) and previous work using the techniques of resource augmentation (e.g. \cite{kalyanasundaram2000speed,koutsoupias1999weak,sleator1985amortized,young1994k,caragiannis2016truthful}), even slightly increased resources usually improves the competitive ratio significantly. However in our problem, we prove that for any constant $c>1$, there exists an instance such that the competitive ratio with augmented resources is $m^{\Omega(1/c)}$, even we allow different item prices for different copies (\Cref{thm: lower_bound_congestion}). In other words, a competitive ratio that is polynomial in $m$ is unavoidable in the tollbooth problem, even if the capacity of each edge is augmented by a constant. We also prove an almost-matching upper bound of $O(m^{1/c})$ in this setting (\Cref{thm: congestion_c}). The upper bound holds for any single-minded welfare maximization problem, where each buyer may demand any set of items instead of edges in a path.

\begin{result}
For any constant integer $c>0$, consider the tollbooth problem where each edge has $c$ copies to sell. There is an instance, such that for any set of prices $\set{p^{k}(i)\mid i\in [m], 1\le k\le c}$ where $p^{k}(i)$ represents the price for the $k$-th copies of edge $i$, the item-pricing mechanism with above prices achieves worst-case welfare an $O(m^{-1/(2c+6)})$-fraction of $\optwel$. On the other hand, there exists an item pricing that achieves worst-case welfare at least an $\Omega(m^{-1/c})$-fraction of $\optwel$.
\end{result}

\subsection{Other Related Work}

\paragraph{Profit maximization for the tollbooth problem.}
The tollbooth problem has been extensively studied in the literature. One line of work~\cite{grigoriev2006sell,elbassioni2009profit,guruswami2005profit,gamzu2010sublogarithmic,kortsarzprofit} aims to efficiently compute prices of items as well as a special subset of buyers called \emph{winners} while maximizing the total profit, such that it is feasible to allocate the demand sets to all winners, and every winner can afford her bundle. There are two major differences between all works above and our setting. Firstly, the seller owns the tie-breaking power in the above works. Secondly and more importantly, in all works above, it is only required that the set of buyers who get their demand sets can afford their demand sets. But there might be other buyers who could afford their demand sets as well but eventually did not get them (or equivalently, not selected as winners).
Since the arriving order is adversarial in our setting, these buyers might come before the winners and take their demand sets. The winners may no longer get their demand sets since some items in their demand sets are already taken.
Therefore, the set of prices computed in the works above may not end up achieving the worst-case welfare equal to the total value of all winners. It is not hard to see that our item-pricing mechanisms are stronger than the settings in the works above: If a set of prices has a competitive ratio $\alpha$ in our setting, then such a set of prices is automatically an $\alpha$-approximation in the setting of the works above, but the converse is not true.

Moreover, the tollbooth problem on star graphs is similar to the graph pricing problem (where prices are given to vertices, and each buyer takes an edge) studied by Balcan and Blum~\cite{balcan2006approximation}, while they considered the unlimited supply setting. They obtained a 4-approximation, which was later shown to be tight by Lee~\cite{lee2015hardness} unless the Unique Games Conjecture is false.
For the multiple and limited supply case, Friggstad et al.~\cite{friggstad2019graph} obtained an $8$-approximation.

\vspace{-.1in}

\paragraph{Walrasian equilibrium for single-minded buyers.} A closely related problem of our setting is the problem of finding market-clearing item prices for single-minded buyers. 
Unlike the resource allocation setting,
a Walrasian equilibrium requires every buyer with a positive utility to be allocated. The existence of the Walrasian equilibrium is proven to be NP-hard, while satisfying $\frac{2}{3}$ of the buyers is possible \cite{chen2004complexity,huang2005approximation,deng2007walrasian,chen2008walrasian}. The hardness of the problem extends to selling paths on graphs, and is efficiently solvable when the underlying graph is a tree \cite{chen2008walrasian}.

\vspace{-.1in}
\paragraph{Profit maximization for single-minded buyers.}
For the general profit maximization for single-minded buyers with unlimited supply, Guruswami et al.~\cite{guruswami2005profit} proved an $O(\log n+\log m)$-approximation. The result was improved to an $O(\log B+\log \ell)$-approximation ratio by Briest and Krysta~\cite{briest2006single}, and then to an $O(\log B)$-approximation by Cheung and Swamy~\cite{cheung2008approximation}. Here $B$ is the maximum number of sets containing an item and $\ell$ is the maximum size of a set. Balcan and Blum~\cite{balcan2006approximation} gave an $O(\ell^2)$-approximation algorithm. Hartline and Koltun~\cite{hartline2005near} gave an FPTAS with a bounded number of items. On the other hand, the problem was proved to be NP-hard for both the limited-supply~\cite{grigoriev2006sell} and unlimited-supply~\cite{guruswami2005profit,briest2006single} case, and even hard to approximate~\cite{demaine2008combination}.

\vspace{-.1in}
\paragraph{Pricing for online welfare maximization with tie-breaking power.} The problem of online resource allocation for welfare maximization has been extensively studied in the prophet inequality literature. In the full-information setting where all buyers' values are known, bundle pricing achieves $2$-approximation to optimal offline welfare \cite{cohen2016invisible}, even when the buyers' values are arbitrary over sets of items. In a Bayesian setting where the seller knows all buyers' value distributions, item pricing achieves a 2-approximation in welfare for buyers with fractionally subadditive values \cite{krengel1978semiamarts,samuel1984comparison,kleinberg2012matroid,feldman2014combinatorial}, and an $O(\log\log m)$-approximation for subadditive buyers \cite{dutting2020log}. For general-valued buyers that demand at most $k$ items, item pricing can achieve a tight $O(k)$-approximation \cite{dutting2020prophet}. \cite{chawla2017stability} studied the problem of interval allocation on a path graph, and achieves $(1-\eps)$-approximation via item pricing when each item has supply $\Omega(k^6/\eps^3)$, and each buyer has a fixed value for getting allocated any path she demands. \cite{chawla2019pricing} further extends the results to general path allocation on trees and gets a near-optimal competitive ratio via anonymous bundle pricing.

\paragraph{Pricing for online welfare maximization without tie-breaking power.} When the seller does not have tie-breaking power, \cite{cohen2016invisible,leme2020computing} show that when there is a unique optimal allocation for online buyers with gross-substitutes valuation functions, static item pricing can achieve the optimal welfare. When the optimal allocation is not unique, \cite{cohen2016invisible,berger2020power} show that a dynamic pricing algorithm can obtain the optimal welfare for gross-substitutes buyers, but for not more general buyers. \cite{hsu2016prices} shows that if the buyers have matroid-based valuation functions, when the supply of each item is more than the total demand of all buyers, the minimum Walrasian equilibrium prices achieve near-optimal welfare. \cite{cohen2016invisible,eden2019max} shows that for an online matching market, when the seller has no tie-breaking power, static item pricing gives at least $0.51$-fraction of the optimal offline welfare, and no more than $\frac{2}{3}$. 


\subsection{Organization}
In Section~\ref{sec:prelim} we describe the settings of the problems studied in the paper in detail. In Section~\ref{sec: path-no-tiebreak}, we present our results on the competitive ratio when the graph is a single path (\Cref{sec:pathnt}) or tree (\Cref{sec:special-graphs-tree-cycle}). 
In Section~\ref{sec: general graph}, we prove upper and lower bounds on the competitive ratio for general graphs and lower bounds for grids.
In Section~\ref{sec:congestion}, we present our results in the setting the capacity of edges in the graph is augmented.
Finally we discuss possible future directions in Section~\ref{sec:future}.

\section{Our Model}\label{sec:prelim}

In this section, we introduce our model in more detail. A seller wants to sell $m$ heterogeneous items to $n$ buyers. Each buyer $j$ is single-minded: She demands a set $Q_j\subseteq[m]$ with a positive value $v_j$.\footnote{Throughout the paper, denote $[m]=\{1,\ldots,m\}$.} Her value for a subset $S\subseteq [m]$ of items is $v_j$ if $Q_j\subseteq S$, and $0$ otherwise. For every buyer $j$, the set $Q_j$ and the value $v_j$ are known to the seller. The seller aims to maximize the welfare, that is, the sum of all buyers' value who get their demand sets. As a special case of the above problem, in the \emph{tollbooth problem}, there is an underlying graph $G$. We denote $V(G)$ and $E(G)$ the vertex and edge set of $G$. Every item in the auction corresponds to an edge in $E(G)$. Let $E(G)=\{e_1,\ldots,e_m\}$. For simplicity, we use the index $i$ to represent the edge $e_i$ as well. For every agent $j$, her demand set $Q_j$ is a single path in graph $G$. For a set of paths $\qset$, denote $E(\qset)=\bigcup_{Q\in \qset}E(Q)$. We say that paths in $\qset$ are \emph{edge-disjoint} (\emph{node-disjoint}, resp.) if all paths in $\qset$ do not share edges (vertices, resp.).  

In the paper we focus on a special class of mechanisms called \emph{item pricing} mechanisms. In an item-pricing mechanism, the seller first computes a posted price $p(e_i)$ (or $p_i$) for every edge $e_i$ in the graph.\footnote{In the paper we allow the posted price $p_i$ to be $\infty$. It means that the price for edge $i$ is sufficiently large, such that no buyer $j$ with $i\in Q_j$ can afford her demand path.} The buyers then arrive one-by-one in some order $\sigma$. When each buyer $j$ arrives, if any edge in her demand set $Q_j$ is unavailable (taken by previous buyers), then she gets nothing and pays 0. Otherwise, she compares her value $v_j$ with the total price $p(Q_j)=\sum_{i\in Q_j}p_i$:

\begin{enumerate}
\item If $p(Q_j)<v_j$, she takes all edges in $Q_j$ by paying $p(Q_j)$; edges in $Q_j$ then become unavailable;
\item If $p(Q_j)>v_j$, she takes nothing and pays $0$;
\item If $p(Q_j)=v_j$, then whether she takes all edges in $Q_j$ at price $p(Q_j)$ depends on the specification about tie-breaking.
\end{enumerate}

We say that the seller has the \emph{tie-breaking power}, if the item pricing mechanism is also associated with a tie-breaking rule. Specifically, whenever $p(Q_j)=v_j$ happens for some buyer $j$, the mechanism decides whether the buyer takes the edges or not, according to the tie-breaking rule. Given any price vector $p=\{p_i\}_{i\in [m]}$ and arrival order $\sigma$, we denote by $\wel(\qset,v; p,\sigma)$ the maximum welfare achieved by the mechanism among all tie-breaking rules. On the other hand, the seller does not have the tie-breaking power (or buyer owns the tie-breaking power) if, whenever $p(S_j)=v_j$ happens for some buyer $j$, the buyer can decide whether she takes the edges or not. For every price vector $p$ and arrival order $\sigma$, we denote by $\welnt(\qset,v; p,\sigma)$ the worst-case (minimum) welfare achieved by the mechanism, over all tie-breaking decisions made by the buyers. In this paper, by default we assume the seller has no tie-breaking power, and will state explicitly otherwise.


For any graph $G$, an instance in this problem can be represented as a tuple $\fset=(\qset,v)=(\{Q_j\}_{j\in [n]},\{v_j\}_{j\in [n]})$ that we refer to as a \emph{buyer profile}. An \emph{allocation} of the items to the buyers is a vector $y\in \set{0,1}^n$, such that for each item $i\in [m]$, $\sum_{j\in [n], i\in Q_j}y_j\le 1$.
Namely, for every $j$, $y_j=1$ if and only if buyer $j$ takes her demand set $Q_j$.
The welfare of an allocation $y$ is therefore $\sum_{j\in [n]}v_jy_j$.
We denote by $\optwel(G,\fset)$ the optimal welfare over all allocations, and use $\optwel$ for short when the instance is clear from the context.

Given any item-pricing mechanism, we define the \emph{competitive ratio} as the ratio of the following two quantities: (i) the offline optimal welfare, which is the total value of the buyers in the optimal offline allocation; and (ii) the maximum among all choices of prices, of the worst-case welfare when the buyers' arrival order $\sigma$ is adversarial. Formally, for any instance $\fset=(\qset,v)$, if the seller does not have tie-breaking power, we define
$$\rationt(\fset)=\frac{\optwel(G,\fset)}{\max_{p}\min_{\sigma}\welnt(\qset,v; p,\sigma)}.$$

In the paper, we analyze the competitive ratio when $G$ has different special structures. For ease of notation, for any graph $G$, denote $\rationt(G)$ the largest competitive ratio $\rationt(\fset)$ for any instance $\fset$ with underlying graph $G$. And given a graph family $\gset$, we denote $\rationt(\gset)=\max_{G\in \gset}\rationt(G)$. For instance, $\rationt(\tree)$ represents the worst competitive ratio among all trees. For the case when the seller has tie-breaking power, we define $\ratio(\qset,v)$, $\ratio(G)$ and $\ratio(\gset)$ similarly.

\section{Competitive Ratio for Special Graphs}
\label{sec: path-no-tiebreak}

In this section, we study the competitive ratio when the underlying graph $G$ is a single path or tree. Although our main focus is on the scenario where buyers own the tie-breaking power, we will start with the setting where $G$ is a single path and the seller owns the tie-breaking power to illustrate the basic idea of how to use the linear program to generate the desired prices. 

\subsection{Warm up: Path Graphs with Tie-Breaking Power}
\label{subsec:path-tie}

Throughout this subsection, we assume that the seller has tie-breaking power. Given any instance $\fset$, the hindsight optimal welfare is captured by an integer program. The relaxed linear program (LP-Primal) and its dual (LP-Dual) are shown as follows.

\begin{minipage}[b]{0.48\textwidth}
\centering
\begin{align*}
\mbox{(LP-Primal)}~~\max &~~\sum_{j\in [n]}v_j\cdot y_j \\
\text{s.t.} &~~\sum_{j: i\in Q_j}y_j\leq 1 &\forall i\in [m]\\
&~~y_j\geq 0 &\forall j\in [n]
\end{align*}
\end{minipage}
\begin{minipage}[b]{0.48\textwidth}
\centering
\begin{align*}
\mbox{(LP-Dual)}~~\min &~~\sum_{i\in [m]}p_i \\
\text{s.t.} &~~\sum_{i\in Q_j}p_i\geq v_j &\forall j\in [n]\\
&~~p_i\ge 0 &\forall i\in [m]
\end{align*}
\end{minipage}

We denote $\optlp(\qset, v)$ (or $\optlp$ if the instance is clear from context) the optimum of (LP-Primal). Clearly, $\optlp(\qset,v)\geq \opt(\qset, v)$. 

The following lemma shows that for any feasible integral solution achieved by rounding from the optimal fractional solution, we are able to compute prices to guarantee selling to the exact same set of buyers via (LP-Dual). The proof is deferred to Appendix~\ref{apd: proof of lem:rounding}.

\begin{lemma}\label{lem:rounding}
Let $y^*$ be any optimal solution of (LP-Primal) and let $y'\in \set{0,1}^n$ be any feasible integral solution of (LP-Primal), such that for each $j\in [n]$, $y'_j=0$ if $y_j^*=0$.
Then there exists a price vector $p=\set{p_i}_{i\in [m]}$ that achieves allocation $y'$ (thus worst-case welfare $\sum_{j\in [n]}v_j\cdot y_j'$), if the seller owns tie-breaking power. 
\end{lemma}

An immediate corollary of Lemma~\ref{lem:rounding} is for the highway problem, i.e. $G$ is a single path. In this case, there is always an integral optimal solution for (LP-Primal)~\cite{cheung2008approximation,schrijver1998theory}. Thus by Lemma~\ref{lem:rounding}, there exist prices that achieve worst-case welfare the same as optimal welfare.

\begin{theorem}
[\cite{cheung2008approximation,chawla2017stability}]
\label{thm:path}
$\ratio(\pathh)=1$.
\end{theorem}

\subsection{Path Graphs without Tie-Breaking Power}\label{sec:pathnt}

In this section, we analyze the competitive ratio for path graphs where the seller has no tie-breaking power. We notice that in the item-pricing mechanism with set of prices $p^*=\{p^*_i\}_{i\in [m]}$ as suggested in \Cref{lem:rounding}, every buyer $j$ with $y_j^*>0$ has 0 utility of buying the path. When buyers own the tie-breaking power, they can make arbitrary decisions and the worst-case welfare may become lower. In \Cref{thm: path_no_tie_lower_bound}, we prove that when the seller has no tie-breaking power, the competitive ratio for path graphs can be strictly larger than 1. The example contains 3 edges $e_1,e_2,e_3$ (from left to right) and 4 buyers. The demand path and value for all buyers are shown in \Cref{table:example_path}. The complete proof is deferred to Appendix~\ref{apd: proof of thm:path_no_tie_lower_bound}. 

\begin{table}
\centering
\begin{tabular}[b]{|c|c|c|c|c|}
\hline
buyer & $1$ & $2$ & $3$ & $4$\\
\hline
path & $e_1$ & $e_3$ & $e_1, e_2$ & $e_2, e_3$\\
\hline
value & $1$ & $1$ & $2$ & $2$\\
\hline
\end{tabular}
\caption{Counterexample for path graph}
\label{table:example_path}
\end{table}

\begin{theorem}
\label{thm: path_no_tie_lower_bound}
$\rationt(\pathh)\geq 3/2$.
\end{theorem}

The main result of this subsection is shown in Theorem~\ref{thm:path_no_tie}, where we prove that the competitive ratio $3/2$ is tight for path graphs.

\begin{theorem}
\label{thm:path_no_tie}
$\rationt(\pathh)\leq 3/2$. 
\end{theorem}

The remainder of this subsection is dedicated to the proof of \Cref{thm:path_no_tie}. According to \Cref{thm:path}, we start with an integral optimal solution of (LP-Primal). Denote $y^*$ the integral optimal solution of (LP-Primal) that maximizes $\sum_{j\in [n]}y^*_j$. 
Define $Y=\set{j\mid y^*_j=1}$ and
$\qset_Y=\set{Q_j\mid y^*_j=1}$. We prove the following lemma, which is useful to guarantee that the constructed price vector is positive in the proof of \Cref{thm:path_no_tie}. The proof is deferred to Appendix~\ref{apd: proof of lem:price-strictly-positive}.


\begin{lemma}\label{lem:price-strictly-positive}
There is an $\eps>0$ and an optimal solution $\set{p^*(e)}_{e\in E(G)}$ for (LP-Dual), such that
(i) for each edge $e\in E(\qset_Y)$, $p^*(e)\geq \eps$; and (ii)
for each $j\in [m]$, either $p^*(Q_j)=v_j$, or $p^*(Q_j)>v_j+m\eps$.
\end{lemma}

Now consider the parameter $\eps>0$ and prices $p^*$ as suggested by Lemma~\ref{lem:price-strictly-positive}. We define $A=\set{j\mid v_j=p^*(Q_j)}$ as the set of buyers who have 0 utility at prices $p^*$. Let $\qset_A=\set{Q_j\mid j\in A}$ be the set of their demand paths. We need the following definition.

\begin{definition}
A set $\qset\subseteq \qset_A$ of edge-disjoint paths is \emph{unique} (in $\qset_A$), if there does not exist a set $\qset'\subseteq\qset_A\setminus \qset$ of $|\qset'|\ge 2$ edge-disjoint paths, such that
$\bigcup_{j:Q_j\in \qset'}Q_j=\bigcup_{j:Q_j\in \qset}Q_j$. 
\end{definition}

Intuitively, a set of edge-disjoint paths $\qset$ is unique if the union of all paths in $\qset$ is a single path or there does not exists another set of paths among the rest paths, whose union is the same as the one of $\qset$. We prove the following lemma for unique edge-disjoint paths. The lemma shows that given any unique set of edge-disjoint paths, we can design proper prices so that the mechanism can serve all buyers whose demand paths are in the set in any arrival order. The proof is deferred to \Cref{apd: proof of claim:covering_implies_optimal}.

\begin{lemma}\label{claim:covering_implies_optimal}
Given any unique set $\qset$ of edge-disjoint paths, there exists a set of positive prices $p=\{p(e)\}_{e\in E(\qset)}$
that achieves worst-case welfare at least $\sum_{Q_j\in \qset}v_j$.
\end{lemma}

\begin{prevproof}{Theorem}{thm:path_no_tie}
We will prove Theorem~\ref{thm:path_no_tie} using \Cref{claim:covering_implies_optimal}. Denote $\qset_Y=\set{Q_1,Q_2,\ldots,Q_k}$, where the paths are indexed according to the order in which they appear on $G$.
First, for each edge $e\notin E(\qset_Y)$, we set its price $\tilde p(e)=+\infty$.
Therefore, any buyer $j$ whose demand path contains an edge not in $E(\qset_Y)$ cannot take her demand path.
In fact, we may assume without loss of generality that $\bigcup_{1\le j\le k}Q_j=G$, since otherwise $\bigcup_{1\le j\le k}Q_j$ is a union of node-disjoint paths and can be divided into separate sub-instances of path graphs. 

The crucial step is to compute three sets $\hat\qset_1,\hat\qset_2,\hat\qset_3\subseteq \qset_A$ of edge-disjoint paths, such that
\begin{enumerate}
\item \label{prop: cover_twice} every edge of $E(G)$ is contained in exactly two paths of $\hat\qset_1,\hat\qset_2,\hat\qset_3$; and
\item \label{prop: good_cover} for each $t\in \set{1,2,3}$, consider every connected component in the graph generated from paths in $\hat\qset_t$. Then set of paths in every connected component is a unique set of edge-disjoint paths.
\end{enumerate}
By property~\eqref{prop: cover_twice} and the fact that $\hat\qset_1,\hat\qset_2,\hat\qset_3\subseteq \qset_A$, $\sum_{j:Q_j\in \hat\qset_1}v_j+\sum_{j:Q_j\in \hat\qset_2}v_j+\sum_{j:Q_j\in \hat\qset_3}v_j=2\cdot \sum_{e\in E(G)}p^*(e)=2\cdot\opt(G,\fset)$. Here the last inequality is because: Since $\bigcup_{1\le j\le k}Q_j=G$, the offline optimal welfare equals to the optimum of (LP-Primal) and the optimum of (LP-Dual).
Assume without loss of generality that $\sum_{j:Q_j\in \hat\qset_1}v_j\ge (2/3)\cdot\opt(G,\fset)$.
Then we can set prices $\set{\tilde p(e)}_{e\in E(\hat\qset_1)}$ according to Lemma~\ref{claim:covering_implies_optimal}, and $+\infty$ price for all other edges. By Lemma~\ref{claim:covering_implies_optimal}, the item pricing $\tilde p$ achieves worst-case welfare at least $\sum_{j:Q_j\in \hat\qset_1}v_j\ge (2/3)\cdot\opt(G,\fset)$.

We now compute the desired sets $\hat\qset_1,\hat\qset_2,\hat\qset_3$ of edge-disjoint paths, which, from the above discussion, completes the proof of Theorem~\ref{thm:path_no_tie}.

We start by defining $\hat\qset$ to be the multi-set that contains, for each path $Q_j\in\qset_Y$, two copies $Q'_j,Q''_j$ of $Q_j$. 
We initially set
\begin{itemize}
\item $\hat\qset_1=\set{Q'_{6r+3},Q'_{6r+4}\mid 0\le r< k/6} \cup\set{Q''_{6r+6},Q''_{6r+1}\mid 0\le r< k/6}$;
\item $\hat\qset_2=\set{Q'_{6r+1},Q'_{6r+2}\mid 0\le r< k/6} \cup\set{Q''_{6r+4},Q''_{6r+5}\mid 0\le r< k/6}$; and
\item $\hat\qset_3=\set{Q'_{6r+2},Q'_{6r+3}\mid 0\le r< k/6} \cup\set{Q''_{6r+5},Q''_{6r+6}\mid 1\le r\le k/6}$.
\end{itemize}
See Figure~\ref{fig:path-1} for an illustration.
Clearly, sets $\hat\qset_1,\hat\qset_2,\hat\qset_3$ partition $\hat\qset$, each contains edge-disjoint paths, and every edge appears twice in paths of $\hat\qset_1,\hat\qset_2,\hat\qset_3$.
However, sets $\hat\qset_1,\hat\qset_2,\hat\qset_3$ may not satisfy Property~\ref{prop: good_cover}.
We will then iteratively modify sets $\hat\qset_1,\hat\qset_2,\hat\qset_3$, such that at the end Property~\ref{prop: good_cover} is satisfied.

Throughout, we also maintain graphs $G_t=\bigcup_{Q\in \hat{\qset}_t}Q$, for each $t\in \set{1,2,3}$. As sets $\hat\qset_1,\hat\qset_2,\hat\qset_3$ change, graphs $G_1,G_2,G_3$ evolve.
We start by scanning the path $G$ from left to right, and process, for each each connected component of graphs $G_1,G_2,G_3$, as follows.

We first process the connected component in $G_1$ formed by the single path $Q_1''$. Clearly, set $\set{Q_1''}$ is unique, since if there are other paths $\hat Q,\hat Q'\in \qset_A$ such that $\hat Q,\hat Q'$ are edge-disjoint and $\hat Q\cup \hat Q'=Q_1$, then the set $\set{\hat Q,\hat Q',Q_2,\ldots,Q_k}$ corresponds to another integral optimal solution $\hat y^*$ of (LP-Primal) with $\sum_{j\in [n]}\hat y^*_j=k+1>k=\sum_{j\in [n]}\hat y^*_j$, a contradiction to the definition of $y^*$. We do not modify path $Q''_1$ in $\hat{\qset}_1$ and continue to the next iteration.

We then process the connected component in $G_2$ formed by the paths $Q_1',Q_2'$. If the set $\set{Q_1',Q_2'}$ is unique, then we do not modify this component and continue to the next iteration. Assume now that the set $\set{Q_1',Q_2'}$ is not unique.
From similar arguments, there exist two other paths $Q^*_1,Q^*_2\in \qset_A$, such that $Q^*_1,Q^*_2$ are edge-disjoint and $Q^*_1\cup Q^*_2=Q_1\cup Q_2$. 
We then replace the paths $Q'_1,Q'_2$ in $\hat{\qset}_2$ by paths $Q^*_1,Q^*_2$.
Let $v^*_1$ be the vertex shared by paths $Q^*_1,Q^*_2$, so $v^*_1\ne v_1$. We distinguish between the following cases.

\textbf{Case 1. $v^*_1$ is to the left of $v_1$ on path $G$.}
As shown in Figure~\ref{fig:path-2}, we keep the path $Q^*_2$ in $\hat{\qset}_2$, and move path $Q^*_1$ to $\hat{\qset}_3$.
Clearly, we create two new connected components: one in $G_3$ formed by a single path $Q^*_1$, and the other in $G_2$ formed by a single path $Q^*_2$.
From similar arguments, the corresponding singleton sets $\set{Q^*_1}, \set{Q^*_2}$ are unique.

\textbf{Case 2. $v^*_1$ is to the right of $v_1$ on path $G$.}
As shown in Figure~\ref{fig:path-4}, we keep the path $Q^*_2$ in $\hat{\qset}_2$, move path $Q^*_1$ to $\hat{\qset}_1$ and additionally move the path $Q'_1$ processed in previous iteration to $\hat{\qset}_2$.
Clearly, we create two new connected components: one in $G_1$ formed by a single path $Q^*_1$, and the other in $G_2$ formed by a single path $Q^*_2$.
From similar arguments, the corresponding singleton sets $\set{Q^*_1}, \set{Q^*_2}$ are unique. 
Note that we have additionally moved $Q'_1$ to $\hat{\qset}_2$, but since we did not change the corresponding component, the singleton set $\set{Q'_1}$ is still unique.

\begin{figure}[h]
\centering
\subfigure[An illustration of paths in set $\hat{\qset}_1\cup \hat{\qset}_2\cup \hat{\qset}_3$ at the beginning.]{\scalebox{0.325}{\includegraphics{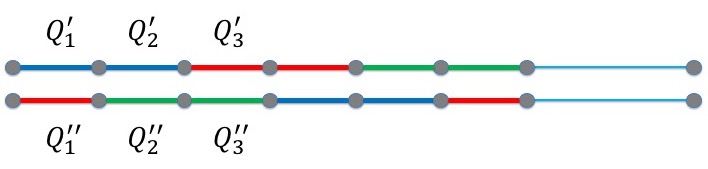}}\label{fig:path-1}}
\subfigure[An illustration of path modification in Case 1.]{\scalebox{0.325}{\includegraphics{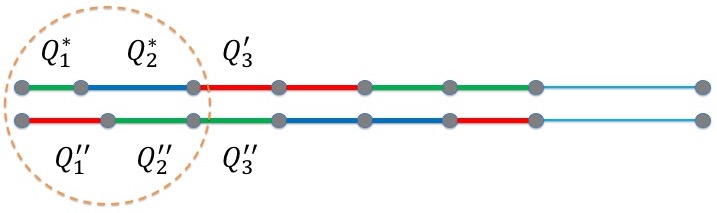}}\label{fig:path-2}}
\subfigure[An illustration of path modification in Case 2.]{\scalebox{0.325}{\includegraphics{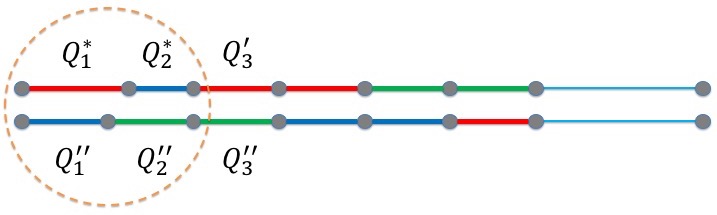}}\label{fig:path-4}}
\subfigure[How old and new paths/components may possibly interact.]{\scalebox{0.325}{\includegraphics{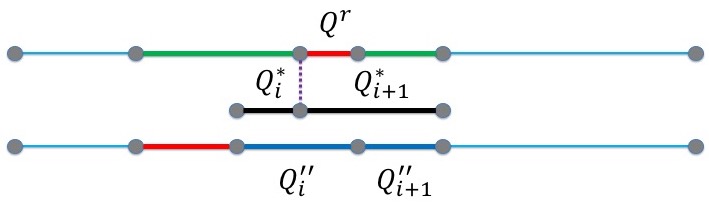}}\label{fig:path-3}}
\caption{Illustrations of the algorithm for computing path sets $\hat{\qset}_1,\hat{\qset}_2,\hat{\qset}_3$.}
\end{figure}

We continue processing the remaining connected components in the same way until all components are unique. We will show that, every time a connected component is not unique and the corresponding two paths are replaced with two new paths, the connected components in $G_1,G_2,G_3$ that we have processed in previous iterations will stay unique. Therefore, the algorithm will end up producing unique components in $G_1,G_2,G_3$ consisting of a unique set of one or two edge-disjoint paths.

To see why this is true, consider an iteration where we are processing a component consisting of paths $Q''_i, Q''_{i+1}$, and there exists edge-disjoint paths $Q^*_i,Q^*_{i+1}$ such that $Q^*_i\cup Q^*_{i+1}=Q''_i\cup Q''_{i+1}$, while the endpoint $v^*_i$ shared by $Q^*_i$ and $Q^*_{i+1}$ is an endpoint of a processed component, as shown in Figure~\ref{fig:path-3}. 
Note that this is the only possibility that the new components may influence the previous components,
However, we will show that this is impossible. 
Note that $Q''_{i}\in \qset_Y$. 
We denote by $Q^r$ the path with endpoints $v^*_i$ and $v_i$, then clearly paths $Q^*_i, Q^r$ are not in $\qset_Y$, edge-disjoint and satisfy that $Q_{i}=Q^*_i\cup Q^r$.
Consider now the set $(\qset_Y\setminus \set{Q_{i}})\cup \set{Q^*_i, Q^r}$. It is clear that this set corresponds to another integral optimal solution $\hat y^*$ of (LP-Primal) with $\sum_{j\in [n]}\hat y^*_j=k+1>k=\sum_{j\in [n]}\hat y^*_j$, a contradiction to the definition of $y^*$.
\end{prevproof}

\subsection{Competitive Ratio for Trees}
\label{sec:special-graphs-tree-cycle}

In this subsection, we study the competitive ratio when graph $G$ is a tree. When seller owns the tie-breaking power, we prove in \Cref{thm:tiebreak-tree-ub-lb} a tight competitive ratio of $\frac{3}{2}$. The upper bound is proved by combining \Cref{lem:rounding} with the integrality gap result of multicommodity problem on tree~\cite{chekuri2007multicommodity, raghavan1994efficient}. On the other hand, we provide an instance on a star to show the competitive ratio is at least $3/2$. We notice that the lower bound also implies that $\ratio(\starr)\geq \frac{3}{2}$ and $\rationt(\tree)\geq \frac{3}{2}$. The proof is deferred to Appendix~\ref{apd: proof of thm:tie-break-tree-ub-lb}.

\begin{theorem}
\label{thm:tiebreak-tree-ub-lb}
$\ratio(\tree)=\frac{3}{2}$.
\end{theorem}

When the seller has no tie-breaking power, we show that the competitive ratio for trees can also be upper bounded by an absolute constant. 

\begin{theorem}\label{thm:tree-no-tiebreak}
For any $\eps>0$, $\rationt(\tree)\leq 7+\eps$.
\end{theorem}

As discussed in the previous subsection for path graphs, the LP-based approach requires the seller to own the tie-breaking power. To prove \Cref{thm:tree-no-tiebreak}, we use a different approach. We first prove the following structural lemma, which partitions a set of edge-disjoint paths into two sets, such that each set of paths forms a union of vertex-disjoint spider graphs. Here a \emph{spider} graph $G$ is a tree with one vertex $u$ of degree at least 3 and all others with degree at most 2. In other words, $E(G)$ can be decomposed into $k$ paths, where any two paths only intersect at $u$. A star is a special spider graph, where all vertices other than $u$ have degree 1.
See \Cref{fig:star_spider} for an example. The proof of \Cref{lem:treetospider} is postponed to \Cref{apd: proof of lem:treetospider}.

\begin{figure}[h]
	\centering
	\subfigure[A star.]{\scalebox{0.08}{\includegraphics{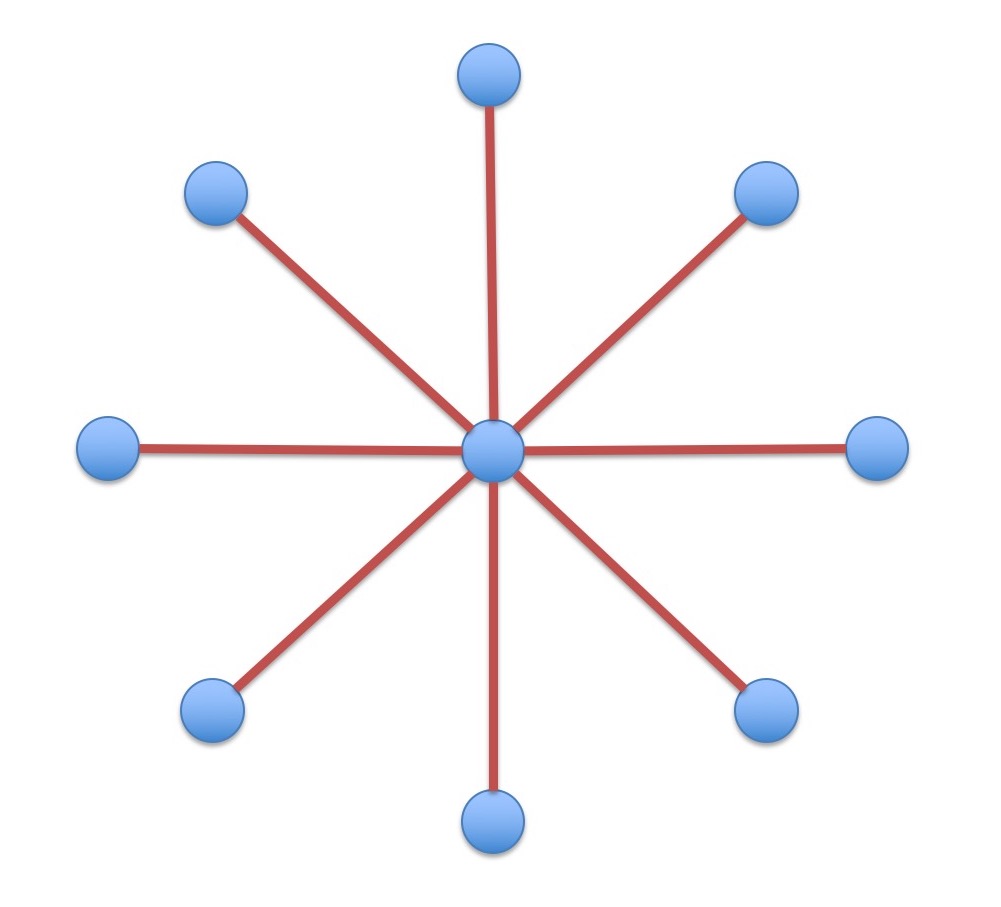}}}
	\hspace{5pt}
	\subfigure[A spider.]{\scalebox{0.08}{\includegraphics{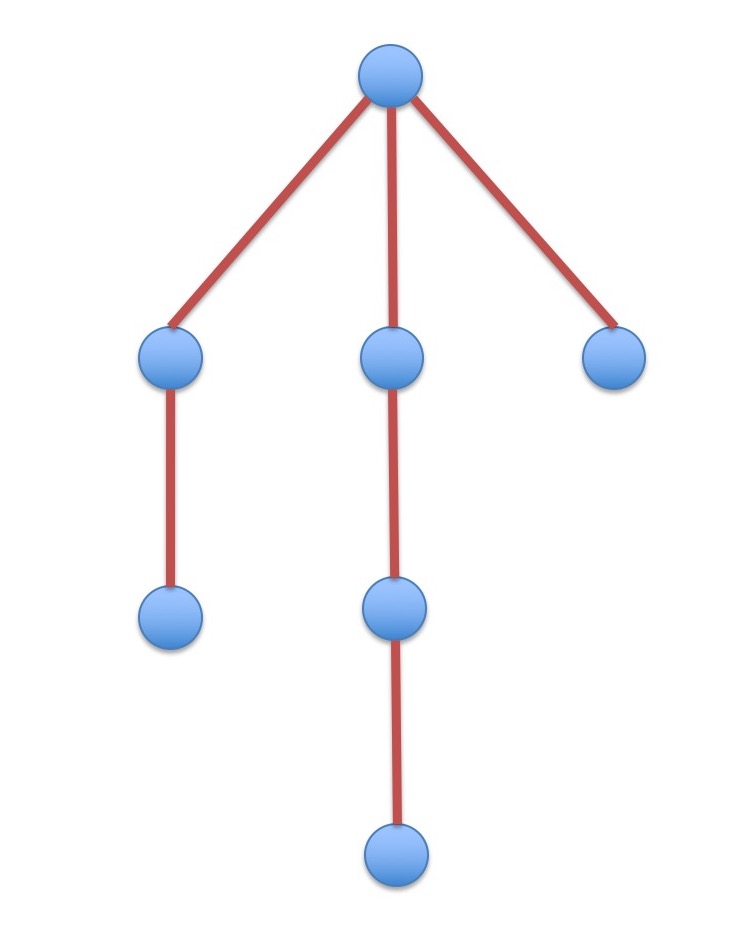}}}
	\caption{An illustration of stars and spiders.\label{fig:star_spider}}
\end{figure}

\begin{lemma}\label{lem:treetospider}
Let $P_1,\ldots,P_n$ be edge-disjoint paths, such that the graph $G=\bigcup_{j\in [n]}P_j$ is a tree, then the set $\pset=\set{P_1,\ldots,P_n}$ can be partitioned into two sets $\pset=\pset'\cup\pset''$, such that both the graph induced by paths in $\pset'$ and the graph induced by paths in $\pset''$ are the union of vertex-disjoint spiders.
\end{lemma}


We prove the following lemma using \Cref{lem:treetospider}. 

\begin{lemma}\label{lem:ratio-treetospider}
$\rationt(\tree)\leq 2\cdot\rationt(\spider)$.
\end{lemma}

\begin{proof}
Given any instance $(G,\mathcal{F})$ on a tree, let $\qset\subseteq \{Q_1,\cdots,Q_m\}$ be the offline optimal solution. It's a set of edge-disjoint paths. Let $\qset', \qset''$ be the partition of $\qset$ according to \Cref{lem:treetospider}. Without loss of generality assume that $\sum_{j:Q_j\in\qset'}v_j\geq \frac{1}{2}\cdot \optwel(G,\fset)$. Consider the item-pricing mechanism with $+\infty$ price on all edges in $\bigcup_{j:Q_j\in \qset''}Q_j$. Then buyers whose demand path contains such an edge can not afford the prices. By the definition of $\qset'$, the remaining buyers form an instance on a union of vertex-disjoint spiders. Thus there exists set of prices whose worst-case welfare is at least $\sum_{j:Q_j\in\qset'}v_j/\rationt(\spider)\geq \optwel(G,\fset)/(2\cdot \rationt(\spider))$.  
\end{proof}

With \Cref{lem:ratio-treetospider} it's sufficient to bound the competitive ratio for spider graphs. 
We take any offline optimal solution. For any path $Q_j$ in the optimal solution that goes through the center of the spider,
we choose the set of prices that achieves more welfare from two strategies: either designing prices to sell $Q_j$ only in the corresponding two legs, or designing prices according to \Cref{thm:path_no_tie} to sell the remaining part of the two legs (which is two separate paths). When $G$ is star, we provide an alternative analysis that induces a better competitive ratio of $2+\eps$. The proof of \Cref{lem:star_and_spider_no_tie} is postponed to \Cref{apd: proof of lem:star_and_spider_no_tie}.

\begin{lemma}\label{lem:star_and_spider_no_tie}
For any $\eps>0$, $\rationt(\spider)\leq 7/2+\eps$. Moreover, $\rationt(\starr)\leq 2+\eps$ for any $\eps>0$.
\end{lemma}

\begin{prevproof}{Theorem}{thm:tree-no-tiebreak}
It directly follows from Lemmas~\ref{lem:ratio-treetospider} and \ref{lem:star_and_spider_no_tie}.
\end{prevproof}




\section{Competitive Ratio for General Graphs}
\label{sec: general graph}  

In this section we study the competitive ratio when $G$ is a general graph. We will start by showing a $poly(m)$ lower bound by constructing an instance in a serial-parallel graph (Section~\ref{subsec: general graph lower bound}). Then we use a modification of this instance to prove lower bounds in grids (Sections~\ref{subsec: grid lower bound}). At last, we prove an upper bound for the competitive ratio in general graphs, that depends on the number of buyers (Section~\ref{sec:generalgraphub}). For results in this section, we assume the strongest dependence on tie-breaking power: the lower bound results hold even when the seller has tie-breaking power, and the upper bound results hold when the seller has no tie-breaking power.    

\subsection{Lower Bound for General Graphs}
\label{subsec: general graph lower bound}

\begin{theorem}\label{thm:lb-mahua}
$\ratio(\emph{\textsf{General\text{ }Graphs}})=\Omega(m^{1/8})$, i.e., there exists a graph $G$ with $|E(G)|=m$ and a buyer profile $\fset$ on $G$, such that no set of prices on edges of $G$ can achieve worst-case welfare $\Omega(\optwel(G,\fset)/m^{1/8})$ even when the seller has tie-breaking power.
\end{theorem}

The remainder of this subsection is dedicated to the proof of Theorem~\ref{thm:lb-mahua}. We will construct the graph $G$ as follows.
For convenience, we will construct a family of graphs $\set{H_{a,b}}_{a,b\in \mathbb{Z}}$, in which each graph is featured by two parameters $a,b$ that are positive integers. We will set the exact parameters used in the proof of \Cref{thm:lb-mahua} later.

For a pair $a,b$ of integers, graph $H_{a,b}$ is defined as follows. The vertex set is $V(H_{a,b})=V_1\cup V_2$, where $V_1=\set{v_0,\ldots,v_{b}}$ and $V_2=\set{u_{i,j}\mid 1\le i\le b,1\le j\le a}$. The edge set is $E(H_{a,b})=\bigcup_{1\le i\le b}E_i$, where $E_i=\set{(v_{i-1},u_{i,j}), (v_{i},u_{i,j})\mid 1\le j\le a}$. 
Equivalently, if we define the multi-graph $L_{a,b}$ to be the graph obtained from a length-$b$ path by duplicating each edge for $a$ times, then we can view $H_{a,b}$ as obtained from $L_{a,b}$ by subdividing each edge by a new vertex.

Let $a,b$ be such that $b=a+3a^3$ and choose $G=H_{a,b}$. Clearly $m=|E(G)|=2ab$, so $a=\Theta(m^{1/4})$ and $b=\Theta(m^{3/4})$. For convenience, we will simply work with graph $L_{a,b}$, since every path in $L_{a,b}$ is also a path in $H_{a,b}$.
Note that $V(L_{a,b})=\set{v_0,\ldots,v_{b}}$, and we denote by $e_{i,1},\ldots,e_{i,a}$ the edges in $L_{a,b}$ connecting $v_{i-1}$ and $v_{i}$.
For brevity, we use the index sequence $(j_1,j_2,\ldots,j_b)$ to denote the path consisting of edges $e_{1,j_1},e_{2,j_2},\ldots,e_{b,j_b}$, where each index $j_t\in [a]$, for each $t\in [b]$. It is clear that a pair of paths $(j_1,j_2,\ldots,j_b)$ and $(j'_1,j'_2,\ldots,j'_b)$ are edge-disjoint iff for each $t\in [b]$, $j_t\ne j'_t$. In the proof we will construct a buyer profile $\fset$ on the multi-graph $L_{a,b}$. Clearly it can be converted to a buyer profile on graph $H_{a,b}$, with the same lower bound of competitive ratio. We prove the following lower bound of competitive ratio for $L_{a,b}$. Theorem~\ref{thm:lb-mahua} follows directly from Lemma~\ref{lem:ratio-Lab} where $a=\Theta(m^{1/4})$ and $b=\Theta(m^{3/4})$.

\begin{lemma}\label{lem:ratio-Lab}
$\ratio(L_{a,b})\geq \sqrt{a}$.
\end{lemma}
\begin{proof}
We define $\fset$ on graph $L_{a,b}$ as follows. We will first define a set $\fset^*$ of buyer profile, and then define, for each subset $S\subseteq [a]$ with $|S|\ge \sqrt{a}$, a buyer profile $\fset_S$, where the buyers in different sets are distinct. Then we define $\fset=\fset^*\cup(\bigcup_{S\subseteq [a], |S|\ge \sqrt{a}}\fset_S)$.

Let the set $\fset^*$ contain, for each $r\in [a]$, a buyer $r$ with $Q^{(r)}=(r,r,\ldots,r)$ and $v_r=1$. 
Clearly, demand paths of buyers in $\fset^*$ are edge-disjoint. In the construction we will guarantee that $\opt(L_{a,b},\fset)=a$, achieved by giving each buyer in $\fset^*$ her demand path.

Before we construct the sets $\set{\fset_S}_{S\subseteq [a], |S|\ge \sqrt{a}}$, we will first state some desired properties of sets $\set{\fset_S}$, and use it to finish the proof of \Cref{lem:ratio-Lab}. Let $\eps>0$ be an arbitrarily small constant.
\begin{enumerate}
	\item \label{prop1}For each $S$, set $\fset_S$ contains $2|S|$ buyers, and every pair $Q,Q'$ of demand paths in $\fset_S$ share an edge. The value for every buyer in $\fset_S$ is $1+\epsilon$.
	\item \label{prop3}For each demand path $Q$ in $\fset_S$, the index sequence $(j^Q_1,\ldots,j^Q_b)$ that  $Q$ corresponds to satisfies that (i) $j^Q_i\in S$ for each $i\in [b]$; and (ii) the set $\set{j^Q_1,\ldots,j^Q_b}$ contains all elements of $S$.
	\item \label{prop4}The union of all demand paths in $\fset_S$ covers the graph $\bigcup_{r\in S}Q^{(r)}$ exactly twice. In other words, for each $i\in [b]$ and every $r\in S$, there are exactly two demand paths $Q$ in $\fset_S$ that satisfy $j_i^Q=r$. 
	\item \label{prop5}For any pair $S,S'\subseteq [a]$ such that $|S|,|S'|\ge \sqrt{a}$ and $S\cap S'\ne \emptyset$, and for any demand path $Q$ in $\fset_S$ and $Q'$ in $\fset_{S'}$, $Q$ and $Q'$ share some edge.
\end{enumerate}
Suppose we have successfully constructed the sets 
$\set{\fset_S}_{S\subseteq [a], |S|\ge \sqrt{a}}$ that satisfy all the above properties. We then define $\fset=\fset^*\cup(\bigcup_{S\subseteq [a], |S|\ge \sqrt{a}}\fset_S)$.
From the above properties, it is easy to see 
that $\opt(L_{a,b},\fset)=a$, which is achieved by giving each buyer in $\fset^*$ her demand path.
We will prove that any prices on edges of $L_{a,b}$ achieve worst-case welfare $O(\sqrt{a})$.

Consider now any set of prices on edges of $L_{a,b}$. We distinguish between the following two cases.

\vspace{-.15in}
\paragraph{Case 1: At least $\sqrt{a}$ buyers in $\fset^*$ can afford their demand paths.} 
We let $S$ be the set that contains all indices $r\in [a]$ such that the buyer $r$ can afford her demand path $Q^{(r)}$, so $|S|\ge \sqrt{a}$.
Consider the set $\fset_S$ of buyers that we have constructed.
We claim that at least some buyers of $\fset_S$ can also afford her demand path. To see why this is true, note that by Property~\ref{prop1}, $\fset_S$ contains $2|S|$ buyers with total value $(2+2\epsilon)|S|$, while the total price of edges in $\bigcup_{r\in S}Q^{(r)}$ is at most $|S|$. Therefore, by Property~\ref{prop4}, there must exist a buyer in $\fset_S$ that can afford her demand path.
We then let this buyer come first and get her demand path $Q$. Then from Property \ref{prop3}, all buyers $r\in S$ can not get their demand paths since their demand paths share an edge with $Q$. All buyers $r\in [a]/S$ can not afford their demand paths. Moreover, for any buyers' arriving order, let $\qset=\{Q_1,...,Q_K\}$ ($Q_1=Q$) be the set of demand paths that are allocated eventually. We argue that $K\leq \sqrt{a}$. For every $k=1,\ldots,K$, $Q_k$ must come from the profile $\bigcup_{S\subseteq [a], |S|\ge \sqrt{a}}\fset_S$. Let $S_k\subseteq [a]$ be the set that $Q_k$ appears in $\fset_{S_k}$. Then by Property \ref{prop1} and \ref{prop5}, we have $S_{k_1}\cap S_{k_2}=\emptyset$ for any $k_1,k_2\in [K],k_1\not=k_2$. Thus we have $K\leq \sqrt{a}$ since $|S_k|\geq \sqrt{a}$ for every $k$.
Hence, the achieved welfare is at most $(1+\eps)\sqrt{a}$, for any buyers' arriving order.
\vspace{-.15in}
\paragraph{Case 2. Less than $\sqrt{a}$ buyers in $\fset^*$ can afford their demand paths.} 
Similar to Case 1, at most $\sqrt{a}$ buyers from sets $\bigcup_{S\subseteq [a], |S|\ge \sqrt{a}}\fset_S$ can get their demand sets simultaneously. Therefore, the total welfare is at most $(1+\eps)\sqrt{a}+\sqrt{a}=O(\sqrt{a})$.

\vspace{.1in}

It remains to construct the sets $\set{\fset_S}_{S\subseteq [a], |S|\ge \sqrt{a}}$ that satisfy all the above properties.
We now fix a set $S\subseteq [a]$ with $|S|\ge \sqrt{a}$ and construct the set $\fset_S$. 
Denote $s=|S|$.
Since each path can be represented by a length-$b$ sequence, we simply need to construct a $(2s\times b)$ matrix $M_S$, in which each row corresponds to a path in $L_{a,b}$.
We first construct the first $s$ columns of the matrix.
Let $N_S$ be an $s\times s$ matrix, such that every row and every column contains each element of $S$ exactly once (it is easy to see that such a matrix exists).
We place two copies of $N_S$ vertically, and view the obtained $(2s\times s)$ matrix as the first $s$ columns of $M_S$.
We then construct the remaining $b-s\geq3a^3$ columns of matrix $M_S$.
Let $S'$ be the multi-set that contains two copies of each element of $S$. 
We then let each column be an independent random permutation on elements of $S'$.
This completes the construction of the matrix $M_S$.
We then add a buyer associated with every path above with value $1+\eps$. This completes the construction of the set $\fset_S$.

We prove that the randomized construction satisfies all the desired properties, with high probability. For Property \ref{prop1}, clearly $\fset_{S}$ contains $2|S|$ buyers, and the value associated with each demand path is $1+\eps$. For every pair of distinct rows in $M_S$, the probability that their entries in $j$-th column for any $j>s$ are identical is $1/2s$, so the probability that the corresponding two paths are edge-disjoint is at most $(1-1/2s)^{3a^3}$.
From the union bound over all pairs of rows, the probability that there exists a pair of edge-disjoint demand paths in $\fset_S$ is at most
$\binom{s}{2}\cdot(1-1/2s)^{3a^3}\le a^2\cdot e^{-a^2}$.
Property \ref{prop3} is clearly satisfied by the first $s$ column of matrix $M_S$.
Property \ref{prop4} is clearly satisfied by the construction of matrix $M_S$.
Therefore, from the union bound on all subsets $S\subseteq[a]$ with $|S|\ge \sqrt{a}$, the probability that Properties \ref{prop1},\ref{prop3}, and \ref{prop4} are not satisfied by all sets $\set{\fset_S}_{S\subseteq [a], |S|\ge \sqrt{a}}$ is $ a^2\cdot e^{-a^2}\cdot 2^{a}<e^{-a}$.

For Property~\ref{prop5}, consider any pair $S,S'$ of such sets and any row in the matrix $M_S$ and any row in the matrix $M_{S'}$.
Since $S\cap S'\neq \emptyset$, the probability that they have the same element in any fixed column is at least $(1/|S|)\cdot(1/|S'|)\geq 1/a^2$, so the probability that the corresponding two paths are edge-disjoint is at most $(1-a^{-2})^{3a^3}<e^{-3a}$. 
From the union bound, the probability that Property~\ref{prop5} is not satisfied is at most $e^{-3a}\cdot (2^a\cdot 2a)^2<e^{-a}$.
Altogether, our randomized construction satisfies all the desired properties with high probability. Thus there must exist a deterministic construction of $\fset_S$ that satisfies all the properties. This completes the proof of \Cref{lem:ratio-Lab}.
\end{proof}

\subsection{Lower Bound for Grids}\label{subsec: grid lower bound}
We notice that in the graph $L_{a,b}$ and $H_{a,b}$ that we constructed in Theorem~\ref{thm:lb-mahua}, the maximum degree among all vertices is $2a$, which is a polynomial of $m$. Readers may wonder if the large polynomial competitive ratio is due to the existence of high-degree vertices that are shared by many demand paths.
In this section, we show a negative answer to this question. We prove that a $poly(m)$ lower bound of the competitive ratio is unavoidable even when $G$ is restricted to be a grid.

\begin{theorem}
\label{thm:bounded-degree-mahua}
Let $G$ be the $(\sqrt{m}\times \sqrt{m})$-grid (so that $G$ has $\Theta(m)$ edges). Then $\ratio(G)=\Omega(m^{1/20})$.
\end{theorem}

The proof is enabled by replacing each high-degree vertex in the graph $H_{a,b}$ with a gadget, so that every vertex in the modified graph $G$ has degree at most 4. Formally, the graph $G=R_{a,b}$ is constructed as follows.
Consider a high-degree vertex $v_i\in V(H_{a,b})$. Recall that it has $2a$ incident edges $\set{(v_i,u_{i,j}),(v_i,u_{i+1,j})\mid j\in [a]}$ in graph $H_{a,b}$ (see \Cref{fig:gadget_before}). The gadget for vertex $v_i$ is constructed as follows. We first place the vertices  $u_{i,1},\ldots,u_{i,a},u_{i+1,a}\ldots,u_{i+1,1}$ on a circle in this order, and then for each $j\in [a]$, we draw a line segment connecting $u_{i,j}$ with $u_{i+1,j}$, such that every pair of these segments intersects, and no three segments intersect at the same point. We then replace each intersection with a new vertex. See Figure~\ref{fig:gadget_after} for an illustration.
\begin{figure}[htb]
\centering
\subfigure[The vertex $v_i$ and its incident edges in $H_{a,b}$. ]{\scalebox{0.3}{\includegraphics{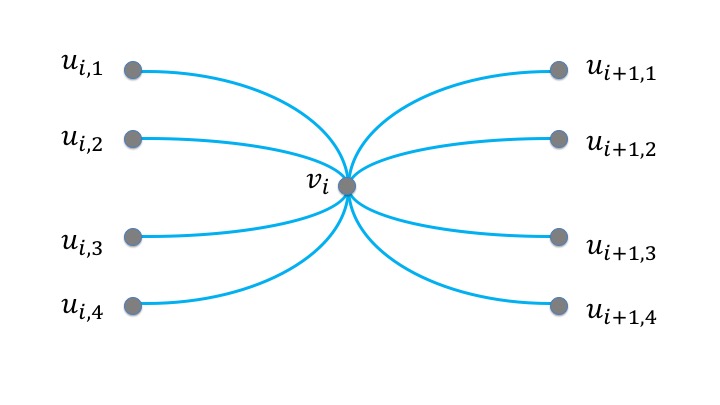}}\label{fig:gadget_before}}
\hspace{1pt}
\subfigure[The gadget graph $K_i$.]{\scalebox{0.3}{\includegraphics{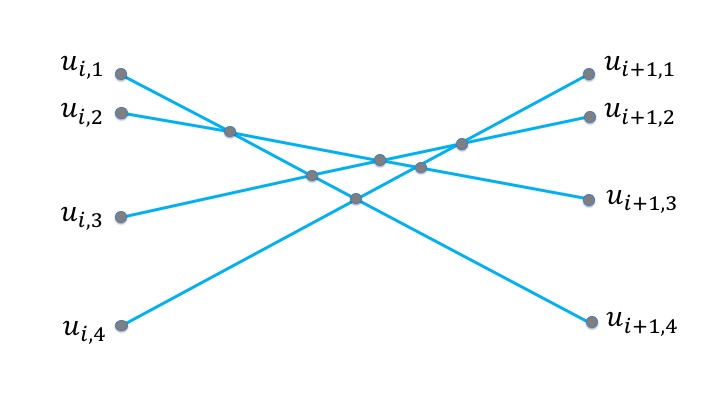}}\label{fig:gadget_after}}
\caption{An illustration of the gadget construction.}\label{fig:gadget}
\end{figure}

Now the modified graph can be embedded in the grid since each vertex has degree at most 4. The complete proof is deferred to Appendix~\ref{apd: proof of thm:bounded-degree-mahua}.





\subsection{Upper Bound for General Graphs}\label{sec:generalgraphub}

At last, we prove the following upper bound for the competitive ratio in any tollbooth problem (on general graphs). The competitive ratio depends on the number of buyers in the auction. Note that when $n$ is sub-exponential on $m$, the competitive ratio proved in \Cref{thm:general_graph} is better than the competitive ratio $O(\sqrt{m})$ proved in~\cite{cheung2008approximation,chin2018approximation}. 

\begin{theorem}
\label{thm:general_graph}
For any given instance $(G,\fset)$ with $|E(G)|=m$ and $|\fset|=n$,  $\rationt(G)=O(m^{0.4}\cdot\log^2m\cdot\log n)$.
\end{theorem}

The complete proof of \Cref{thm:general_graph} can be found in Appendix~\ref{apd: proof of thm:general_graph}. Here we only provide a sketch of how we construct the prices. Take any offline optimal solution $\qset$. As a pre-processing, we first select a subset $\qset'$ of $\qset$ such that each path in the subset has both length and value within 2 times of other paths, by losing a ratio of $\log^2 m$. Then we increase the prices for each edge in two steps. In the first step, let $\qset''$ be a random subset of $\qset'$ by including each path independently with probability $1/2$. We set the price for each edge not contained in any path of $\qset''$ to be $\infty$. In the second step, we pick a special set of short paths and increase the prices of each edge uniformly for every path in the set. With the above prices, we are able to prove a lower bound of the size of any set of paths sold in the selling process, which contributes enough welfare compared to the $\optwel$.

\section{Resource Augmentation}\label{sec:congestion}
In Sections~\ref{sec: path-no-tiebreak} and \ref{sec: general graph}, we studied the competitive ratio of the tollbooth problem, in which each edge can be allocated to at most one buyer. In this section, we consider the case where each item has augmented resources. We prove results in general combinatorial auctions with single-minded buyers, which is a generalization of the tollbooth problem. In a combinatorial auction, let $U=\set{1,\ldots,m}$ be the item set. Given a buyer profile $\fset=\set{(Q_j,v_j)}_{j\in [n]}$, we denote by $\opt(U,\fset)$ the maximum welfare by allocating items in $U$ to the buyers, such that each item is assigned to at most one buyer.
The seller, however, has more resources to allocate during the selling process. For each item $i\in U$, the seller has $c$ copies of the item, and each copy is sold to at most one buyer. In an item-pricing mechanism, the seller is allowed to set different prices for different copies. Formally, for each item $i\in U$, the seller sets $c$ prices $p^1(i)\leq\ldots\leq p^c(i)$, such that for each $1\le k\le c$, the $k$-th copy of item $i$ is sold at price $p^k(i)$.
When a buyer comes, if $k-1$ copies of item $i$ has already been sold, the buyer can purchase item $i$ with price $p^k(i)$. Again we define the worst-case welfare of an item pricing as the minimum welfare among all the buyers' arriving order.

With the augmented resources, the seller can certainly achieve more welfare than in the case with a single unit per item. 
We show that item pricing can achieve worst-case welfare $\Omega(m^{-1/c})\cdot\opt(U,\fset)$. The result implies that when $c=\Omega(\log m)$, item pricing on $c$ copies achieves at least a constant factor of the offline optimal welfare when each edge has supply 1.
Due to space limit, the proofs of all theorems in the section are deferred to Appendix~\ref{apd: missing sec:congestion}. 
\begin{theorem}
\label{thm: congestion_c}
For any buyer profile $\fset=\set{(Q_j,v_j)}_{j\in [n]}$ and any integer $c>0$, there exists a set\\ $\set{p^{k}(i)\mid i\in U, 1\le k\le c}$ of prices on items of $U$, that achieves worst-case welfare $\Omega(m^{-1/c})\cdot\opt(U,\fset)$, even when the seller has no tie-breaking power.
\end{theorem}


On the other hand, in \Cref{thm: congestion_c lower bound} we show that a polynomial dependency on $m$ in the competitive ratio is in fact unavoidable. 

\begin{theorem}
\label{thm: congestion_c lower bound}
For any integer $c>0$, there exists a ground set $U$ with $|U|=m$ and a buyer profile $\fset=\set{(Q_j,v_j)}_{j\in [n]}$, such that any item-pricing mechanism with set of prices $\set{p^{k}(i)\mid i\in U, 1\le k\le c}$ achieves worst-case welfare $O(c\cdot m^{-1/(c+1)})\cdot \opt(U,\fset)$, even when the seller has tie-breaking power.
\end{theorem}

In \Cref{thm: lower_bound_congestion}, we prove that a polynomial welfare gap also exists in the tollbooth problem. We adapt the series-parallel graph $H_{a,b}$ used in \Cref{thm:lb-mahua} and show that a polynomial competitive ratio is unavoidable in the tollbooth problem, even each edge has a constant number of copies. 
\begin{theorem}
\label{thm: lower_bound_congestion}
In the tollbooth problem, for any constant integer $c>0$, there exists a graph $G$ with $m$ edges and a buyer profile $\fset=\set{(Q_j,v_j)}_{j\in [n]}$, such that any set $\set{p^{k}(e)\mid e\in E(G), k\in [c]}$ of prices achieves worst-case welfare $O(m^{-1/(2c+6)})\cdot \optwel(G,\fset)$.
\end{theorem}

\section{Future Work}\label{sec:future}
We study the worst-case welfare of item pricing in the tollbooth problem. 
There are several future directions following our results. Firstly, in the paper we assume that all buyers' value are all public. A possible future direction is to study the Bayesian setting where the seller does not have direct access to each buyer's value, but only know the buyers' value distributions. Secondly, we focus on the tollbooth problem where each buyer demands a fixed path on a graph. An alternative setting is that each buyer has a starting vertex and a terminal vertex on the graph, and she has a fixed value for getting routed through any path on the graph. Such a setting is equivalent to our problem when the underlying graph is a tree, where a constant competitive ratio is proved in this paper even if the seller does not have tie-breaking power. 
However, there may exist more than one path between two vertices in a graph with cycles, and thus the buyer is not single-minded in this setting. In the paper we have shown that item pricing may not approximate the optimal welfare well in the tollbooth problem. It remains open whether the item pricing performs well in the alternative setting.  
Thirdly, the power of tie-breaking hasn't been studied much in the literature on mechanism design. In this paper we show that the tie-breaking power may cause a difference in the tollbooth problem even when the graph is a single path. It would be interesting to see other scenarios where the tie-breaking power also makes much difference.

\bibliographystyle{alpha}
\bibliography{reference}

\appendix

\section{Missing Details from Section~\ref{sec: path-no-tiebreak}}
\label{apd: missing sec: path-no-tiebreak}

\subsection{Proof of Lemma~\ref{lem:rounding}}
\label{apd: proof of lem:rounding}

\begin{proof}
Let $p^*=\{p^*_i\}_{i\in [m]}$ be any optimal solution for (LP-Dual). 
We will show that, for any buyers' arrival order $\sigma$, with a proper tie-breaking rule, we can ensure with the prices $p^*$ that the outcome of the selling process is $y'$, and therefore it achieves welfare $\sum_{j\in [n]}v_j\cdot y_j'$.

Consider now a buyer $j$.
Assume first that $y_j'=0$, from the constraint of (LP-Dual), $\sum_{i\in Q_j}p^*_i\geq v_j$.
If $\sum_{i\in Q_j}p^*_i> v_j$, then the buyer $j$ cannot afford her demand set.
If $\sum_{i\in Q_j}p^*_i= v_j$, since the seller has the power of tie-breaking, the seller can choose not to give the demand set to her.
Assume now that $y_j'=1$, so $y_j^*>0$. From the complementary slackness, $\sum_{i \in Q_j}p^*_i=v_j$. The seller can sell the demand set $Q_j$ to buyer $j$.
Since $y'$ is a feasible solution of (LP-Primal), no items will be given to more than one buyer.
Moreover, it is easy to verify that the above tie-breaking rule ensures that this selling process's outcome is exactly $y'$, for any buyers' arrival order $\sigma$. Therefore, the worst-case welfare is $\sum_{j\in [n]}v_j\cdot y_j'$. The set of prices $p^*$ can be computed in polynomial time by solving the (LP-Dual).
\end{proof}

\subsection{Proof of Theorem~\ref{thm: path_no_tie_lower_bound}}
\label{apd: proof of thm:path_no_tie_lower_bound}
\begin{proof}
Let $G$ be a path consisting of three edges $e_1,e_2,e_3$, that appears on the path in this order. The instance is described below.
\begin{itemize}
\item Buyer 1 demands the path consisting of a single edge $e_1$, with value $1$;
\item Buyer 2 demands the path consisting of a single edge $e_3$, with value $1$;
\item Buyer 3 demands the path consisting of edges $e_1,e_2$, with value $2$; and
\item Buyer 4 demands the path consisting of edges $e_2,e_3$, with value $2$.
\end{itemize}
	
It is clear that the optimal allocation is assigning $e_1$ to buyer 1 and $e_2,e_3$ to buyer 4 (or assigning $e_3$ to buyer 2 and $e_1,e_2$ to buyer 3). The optimal welfare $\opt(G,\fset)=3$.
We now show that, for any prices $p_1,p_2,p_3$, there is an order $\sigma$ of the four buyers such that the obtained welfare is at most $2$, when the seller has no tie-breaking power. We distinguish between two cases.

\textbf{Case 1. $p_1\leq 1$ and $p_3\leq 1$.} We let buyers 1 and 3 come first and take edges $e_1$ and $e_3$ respectively. Note that the seller does not have the tie-breaking power. The buyer can decide whether or not to take the path when the price equals the value. Now buyers 3 and 4 cannot get their paths. So the obtained welfare is $2$.

\textbf{Case 2. $p_1>1$ or $p_3>1$.} Assume without loss of generality that $p_1>1$. If $p_2+p_3\leq 2$, then we let buyer 4 come first and take edges $e_2$ and $e_3$. It is clear that no other buyer can get her path. So the total welfare is 2. If $p_2+p_3>2$, then $p_1+p_2+p_3>3$. Note that it is impossible in this case all edges are sold since the total price of all edges is larger than the hindsight optimal welfare 3. Therefore, the obtained welfare is at most $2$ as at most two edges are sold.
\end{proof}

\subsection{Proof of \Cref{lem:price-strictly-positive}}
\label{apd: proof of lem:price-strictly-positive}
\begin{proof}
Let $\eps'=\frac{1}{m^2}\cdot\min_jv_j$. We construct another instance $\fset'$ from $\fset$ as follows. We add, for each $e\in E(G)$ a new buyer demanding the path consisting of a single edge $e$ with value $\eps'$.
On one hand, it is easy to see that $y^*$ is still an optimal solution of this new instance. Let $p^*$ be any optimal solution for the corresponding dual LP for the new instance $\fset$. Define $\eps''=\frac{1}{2m}\min_{j:p^*(Q_j)>v_j}(p^*(Q_j)-v_j)$ and $\eps=\min\{\eps',\eps''\}$. Clearly $p^*$ is also an optimal dual solution for the instance $\fset$ and satisfies both properties.
\end{proof}

\subsection{Proof of \Cref{claim:covering_implies_optimal}}
\label{apd: proof of claim:covering_implies_optimal}

\begin{proof}
We denote $I=\bigcup_{j:Q_j\in \qset}Q_j$ and $\qset_I=\{Q_j\in \qset_A\mid Q_j\subseteq I\}$.
First, for each edge $e\notin E(I)$, we set its price $p(e)=+\infty$.
We will show that we can efficiently compute prices $\set{p(e)}_{e\in E(I)}$ for edges of $I$, such that (i) for each path $Q_j\in \qset$, $v_j>p(Q_j)$; (ii) for each path $Q_j\in \qset_I\setminus \qset$, $v_j<p(Q_j)$; and (iii) $|p(e)-p^*(e)|\leq \eps$, for every $e\in E(I)$. 

Consider the item pricing $p$ which satisfies all three properties. By property (iii) and Lemma~\ref{lem:price-strictly-positive}, any buyer $j\not\in A$ can not afford her path, since $p(Q_j)\geq p^*(Q_j)-n\eps>v_j$.
It is clear that the set $\set{p(e)}_{e\in E(G)}$ of prices achieves worst-case welfare $\sum_{Q_j\in \qset}v_j$. For buyers in $A$, according to the first two properties, only buyers $j$ where $Q_j\in \qset$ can purchase their demands. Thus the worst-case welfare is $\sum_{j:Q_j\in \qset}v_j$.

The existence of prices $\set{p(e)}_{e\in E(I)}$ that satisfy (i), (ii) is equivalent to the feasibility of the following system.
\begin{equation}\label{lp:covering-primal}
\left\{
\begin{array}{cccc}
\displaystyle \sum_{e\in Q_j}p(e) &< & v_j,&\forall Q_j\in\qset;\\
\displaystyle \sum_{e\in Q_j}p(e) & >& v_j,&\forall Q_j\in \qset_I\setminus \qset.
\end{array}
\right.
\end{equation}
From the definition of $A$, and since $\qset_I\subseteq\qset_A$, for each path $Q_j\in \qset_I$, $\sum_{e\in Q_j}p^*(e)= v_j$.
We denote $\alpha(e)=p(e)-p^*(e)$ for all $e\in E(I)$, then system~\eqref{lp:covering-primal} is feasible if the following system is feasible, for some small enough $\eps'>0$.
\begin{equation}\label{lp:covering-primal2}
\left\{
\begin{array}{cccc}
\displaystyle \sum_{e\in Q_j}\alpha(e) &\le & -\eps',&\forall Q_j\in\qset;\\
\displaystyle \sum_{e\in Q_j}\alpha(e) & \ge & \eps',&\forall Q_j\in \qset_I\setminus \qset.
\end{array}
\right.
\end{equation}
From Farkas' Lemma,
system~\eqref{lp:covering-primal2} is feasible if and only if the following system is infeasible.
\begin{equation}\label{lp:covering-dual}
\left\{
\begin{array}{rccc}
\displaystyle \sum_{j:Q_j\in \qset}\beta_j\cdot\mathbbm{1}[e\in Q_j]-\sum_{j:Q_j\in \qset_I\setminus \qset}\beta_j\cdot\mathbbm{1}[e\in Q_j] &=& 0,&\forall e\in E(I);\\
\displaystyle\sum_{j:Q_j\in \qset_I}\beta_j&>& 0;\\
\beta_j & \in & [0,1],&\forall Q_j\in \qset_I.
\end{array}
\right.
\end{equation}
where the additional constraints $\beta_j\le 1, \forall Q_j\in \qset_I$ will not influence the feasibility of the system due to scaling.


Next we will prove that system \eqref{lp:covering-dual} is feasible iff it admits an integral solution.

\begin{definition}
We say that a matrix $A$ is \emph{totally unimodular} iff the determinant of every square submatrix of $A$ belongs to $\{0,-1,+1\}$.
\end{definition}
\begin{lemma}[Forklore]
\label{lem:forklore}
If a matrix $A\in \mathbb{R}^{m\times n}$ is totally unimodular and a vector $b\in\mathbb{R}^m$ is integral, then every extreme point of the polytope $\{x\in \mathbb{R}^n\mid Ax\leq b, x\geq 0\}$ is integral.
\end{lemma}

\begin{lemma}[\cite{schrijver1998theory}] \label{lem:proposition-tu}
A matrix $A\in \mathbb{R}^{m\times n}$ is totally unimodular iff for any subset $R\subseteq [n]$, there exists a partition $R=R_1\cup R_2$, $R_1\cap R_2=\emptyset$, such that for any $j\in [m]$,
$$\sum_{i\in R_1}a_{ij}-\sum_{i\in R_2}a_{ij}\in \{0,-1,+1\}.$$
\end{lemma}

\begin{lemma}\label{lem:covering-integral-solu}
System~\eqref{lp:covering-dual} is feasible if and only if it admits an integral solution.
\end{lemma}
\begin{proof}
We prove that the coefficient matrix $C$ in system \eqref{lp:covering-dual} is totally unimodular.\footnote{The second inequality in system \eqref{lp:covering-dual} is strict inequality. We can modify it to an equivalent inequality $\sum_{j:Q_j\in \qset}\beta_j\geq \eps_0$ for some sufficiently small $\eps_0>0$ and then apply Lemma~\ref{lem:forklore}.} For any set of rows $R=\{i_1,i_2,\cdots,i_k\}$ such that $i_1<i_2<\cdots<i_k$, We define the partition $R_1=\{i_1,i_3,i_5,...\}$ be the set with odd index and $R_2=\{i_2,i_4,i_6,...\}$ be the set with even index. For every column $j$, $Q_j$ is an subpath and thus the ones (or -1s for $Q_j\in \qset_I\setminus \qset$) in the vector $(C_{ij})_{i=1}^n$ are consecutive. Hence the difference between $\sum_{i\in R_1}C_{ij}$ and $\sum_{i\in R_2}C_{ij}$ is at most 1. By Lemma~\ref{lem:proposition-tu}, $C$ is totally unimodular.
\end{proof}

Assume for contradiction that  \eqref{lp:covering-dual} is feasible, let $\beta^*\in \set{0,1}^{\qset_I}$ be an integral solution of \eqref{lp:covering-dual}.
Note that the constraints
$\sum_{j:Q_j\in \qset}\beta^*_j\cdot\mathbbm{1}[e\in Q_j]-\sum_{j:Q_j\in \qset_I\setminus \qset}\beta^*_j\cdot\mathbbm{1}[e\in Q_j] = 0,\forall e\in E(I)$ and the fact that paths of $\qset$ are edge-disjoint
imply that the paths in set $\qset'=\{Q_j\mid \beta_j=1\}$ are edge-disjoint. Note that $\qset'\cap\qset=\emptyset$ and $\bigcup_{Q_j\in \qset'}Q_j=\bigcup_{Q_j\in \qset}Q_j$, and this is a contradiction to the assumption that the set $\qset$ is good.
Therefore, system \eqref{lp:covering-dual} is infeasible, and thus system \eqref{lp:covering-primal2} is feasible.
Let $(\alpha(e))_{e\in E(I)}$ be a solution of system~\eqref{lp:covering-primal2}, such that $\sum_{e\in E(I)} |\alpha(e)|\le \eps$. It is clear that such a solution exists due to scaling.
We set $p(e)=p^*(e)+\alpha(e)$ for all $e\in E(I)$, and it is clear that the prices $\set{p(e)}_{e\in E(G)}$ satisfy all the properties of \Cref{claim:covering_implies_optimal}.
Moreover, prices $\set{p(e)}_{e\in E(G)}$ are positive.
\end{proof}

\section{Missing Details from Section~\ref{sec:special-graphs-tree-cycle}}
\label{apd: missing sec:special-graphs-tree-cycle}

\subsection{Proof of Theorem~\ref{thm:tiebreak-tree-ub-lb}}
\label{apd: proof of thm:tie-break-tree-ub-lb}

\begin{proof}

On one hand, we can upper bound the competitive ratio as a corollary of Lemma~\ref{lem:rounding} is for the setting where the underlying graph $G$ is a tree. Notice that the allocation problem can be viewed as a multicommodity problem on tree, with the demand of each source-destination pair and the supply of each edge being 1. \cite{chekuri2007multicommodity,raghavan1994efficient} shows that for such an instance, the integrality gap for rounding the primal linear program is essentially $\frac{2}{3}$. Thus by Lemma~\ref{lem:rounding}, $\ratio(\tree)\leq \frac{3}{2}$.

On the other hand, we provide an example showing that the $3/2$ competitive ratio is tight for star graphs, i.e trees of height 1 with all but one node being leaves. Consider a star graph with four edges $e_1,e_2,e_3,e_4$.
	The instance $\fset$ is described as follows:
	\begin{itemize}
		\item Buyer 1 demands the path $Q_1$ consisting of edges $e_1, e_2$, with value $v_1=1$;
		\item Buyer 2 demands the path $Q_2$ consisting of edges $e_1, e_3$, with value $v_2=2+\eps$;
		\item Buyer 3 demands the path $Q_3$ consisting of edges $e_1,e_4$, with value $v_3=2+\eps$;
		\item Buyer 4 demands the path $Q_4$ consisting of edges $e_3,e_4$, with value $v_4=2$.
	\end{itemize}
	It is clear that the optimal allocation is assigning the edges $e_1,e_2$ to buyer 1, and assigning the edges $e_3,e_4$ to buyer 4. The optimal welfare is $\optwel(G,\fset)=3$. However, we will show that any set of prices on the edges achieve the worst-case welfare at most $2+\eps$. Then $\ratio(\text{Star})\geq \frac{3}{2+\eps}>\frac{3}{2}-\eps$.
	
	Let $\set{p_i}_{1\le i\le 4}$ be any set of prices.
	If $p(Q_1)>1$, then buyer 1 cannot afford her demand path $Q_1$. Since each pair of paths in $Q_2, Q_3, Q_4$ shares an edge, at most one of buyers $2,3,4$ may get her demand path, so the welfare is at most $2+\eps$. 
	If $p(Q_4)>2$, then via similar arguments we can show that the welfare is at most $2+\eps$. 
	Consider now the case where $p(Q_1)=p_1+p_2\le 1$ and $p(Q_4)=p_3+p_4\le 2$. We have $p_1\le 1$, and one of $p_3,p_4$ is at most $1$. Assume w.l.o.g. that $p_3\le 1$. We let buyer 3 comes first. 
	Note that $p(Q_3)=p_1+p_3\le 2<2+\eps$. Therefore, buyer 3 will get her demand path, and later on no other buyer may get her path. Altogether, the worst-case welfare for any set of prices is at most $2+\eps$. 
Combine the two bounds above we prove the correctness of Theorem~\ref{thm:tiebreak-tree-ub-lb}.
\end{proof}

\subsection{Proof of Lemma~\ref{lem:treetospider}}
\label{apd: proof of lem:treetospider}
\begin{proof}
We choose arbitrarily a vertex $r\in V(G)$ as designate it as the root of the tree.
We iteratively construct a sequence of subsets of  $\pset$ as follows.
We first let $\pset_1$ contains all paths of $\pset$ that contains the vertex $r$.
Then for each $i>1$, as long as $\pset\ne \bigcup_{1\le t\le i-1}\pset_t$, we let $\pset_i$ contains all paths in $\pset\setminus \left(\bigcup_{1\le t\le i-1}\pset_t\right)$ that shares a vertex with a path in $\pset_{i-1}$.
We continue this process until all paths of $\pset$ are included in sets $\pset_1,\ldots,\pset_k$ for some $k$.
See Figure~\ref{fig:spider_decomp} for an illustration.
We then let $\pset'=\bigcup_{0\le t\le \lfloor{k/2}\rfloor}\pset_{2t+1}$ and $\pset''=\bigcup_{1\le t\le \lfloor{k/2}\rfloor}\pset_{2t}$. This completes the construction of the sets $\pset',\pset''$. Clearly, sets $\pset'$ and $\pset''$ partition $\pset$.

\begin{figure}[ht]
\centering
\subfigure[The graph $G$. Paths are shown in distinct colors.]{\scalebox{0.25}{\includegraphics{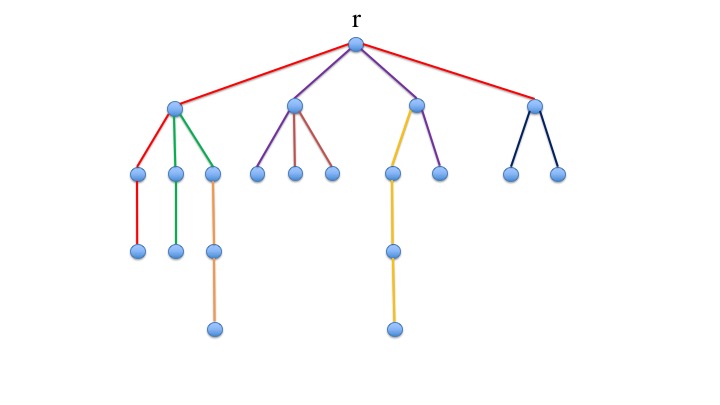}}\label{fig:non_crossing_representation}}
\hspace{1pt}
\subfigure[The paths in $\pset_1$ are shown in dash lines.]{\scalebox{0.25}{\includegraphics{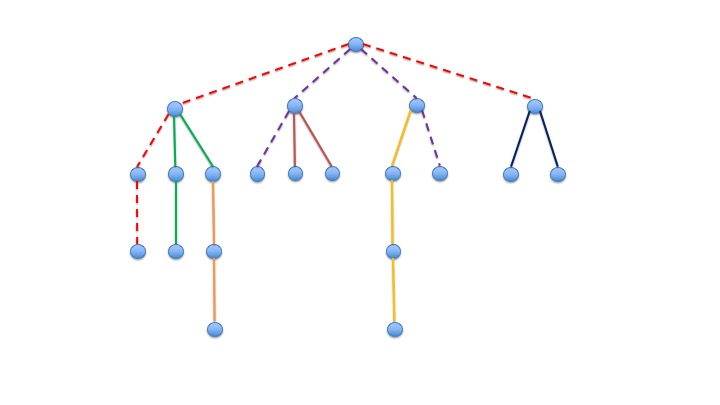}}}
\subfigure[The paths in $\pset_2$ are shown in dash lines.]{\scalebox{0.25}{\includegraphics{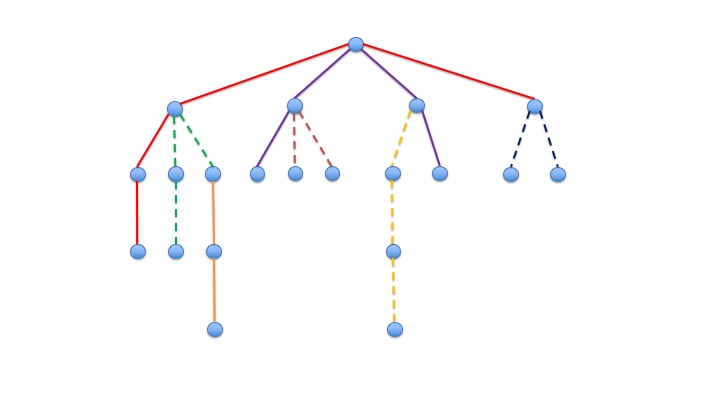}}}
\hspace{1pt}
\subfigure[The paths in $\pset_3$ are shown in dash lines.]{\scalebox{0.25}{\includegraphics{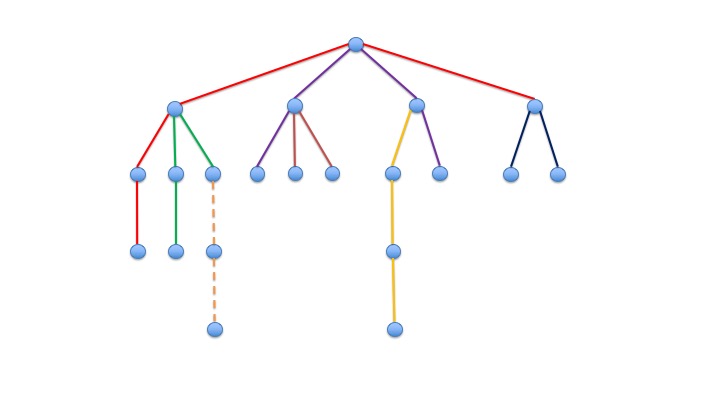}}}
\caption{An illustration of decomposing paths into subsets $\pset_1,\pset_2,\pset_3$.\label{fig:spider_decomp}}
\end{figure}

It remains to show that the graph induced by the paths in $\pset'$ is the union of vertex-disjoint spiders. The proof for $\pset''$ is similar.
First, for $0\le t<t'\le \lfloor{k/2}\rfloor$, we show that the paths in $\pset_{2t+1}$ are vertex-disjoint from the paths in $\pset_{2t'+1}$. Assume for contradiction that this is false, and let $P\in \pset_{2t+1}, P'\in \pset_{2t'+1}$ be a pair of paths that share a common vertex. However, according to the process of constructing the sets $\pset_1,\ldots,\pset_k$, if $P$ is included in $\pset_{2t+1}$ while $P'$ is not included in any set of $\pset_1,\ldots,\pset_{2t+1}$, since $P'$ shares a vertex with $P'$, $P'$ should be included in the set $\pset_{2t+2}$ rather than $\pset_{2t'+1}$, a contradiction.
We then show that, for each $0\le t\le \lfloor{k/2}\rfloor$, the paths in $\pset_{2t+1}$ form disjoint spiders. Clearly the paths in $\pset_1$ form a spider, since they are edge-disjoint paths that share only the root $r$ of tree $G$. Consider now some $t\ge 1$.
Let $v,v'$ be any two distinct vertices of $V(\pset_{2t})\cap V(\pset_{2t+1})$. We denote by $\pset_{2t+1}(v)$ ($\pset_{2t+1}(v')$, resp.) the subset of paths in $\pset_{2t+1}$ that contains the vertex $v$ ($v'$, resp.).
We claim that the paths in $\pset_{2t+1}(v)$ are vertex-disjoint from the paths in $\pset_{2t+1}(v')$.
Note that, if the claim is true, since the paths in $\pset_{2t+1}(v)$ only shares a single node $v$ (as otherwise there is a cycle caused by some pair of paths in $\pset_{2t+1}(v)$, a contradiction to the fact that graph $G=\bigcup_{j\in [n]}P_i$ is a tree), it follows that the paths in $\pset_{2t+1}(v)$ form a spider, and altogether, the paths in $\pset'$ form vertex-disjoint spiders.

It remains to prove the claim.
Assume that this is false. Let $P_{2t+1}\in \pset_{2t+1}(v), P'_{2t+1}\in \pset_{2t+1}(v')$ be a pair of paths that shares a common vertex $u$, with $u\ne v,v'$. Note that, from our process of constructing the sets $\pset_1,\ldots,\pset_k$, there are two sequences of paths $P_1,P_2,\ldots,P_{2t}$ and $P'_1,P'_2,\ldots,P'_{2t}$, such that
(i) for each $1\le s\le 2t$, $P_{s}, P'_{s}\in \pset_s$;
(ii) for each $1\le s\le 2t$, $P_{s}$ shares a vertex with $P_{s+1}$, and $P'_{s}$ shares a vertex with $P'_{s+1}$.
Since $P_1$ and $P'_1$ shares the root $r$ of tree $G$, it follows that the graph consisting of paths in $\set{P_s,P'_s \mid 1\le s\le 2t+1}$ contains a cycle, a contradiction to the fact that $G$ is a tree.
\end{proof}

\subsection{Proof of Lemma~\ref{lem:star_and_spider_no_tie}}
\label{apd: proof of lem:star_and_spider_no_tie}

For star and spider graphs, we prove that the competitive ratios are at most $2$ and $7/2$ respectively. 

\begin{theorem}
\label{thm:star}
For any $\eps>0$, $\rationt(\starr)\leq 2+\eps$. 
\end{theorem}

\begin{proof}[Proof of Theorem~\ref{thm:star}]
Let $\qset$ be the set of paths in the optimal allocation, so $\optwel=\sum_{j:Q_j\in \qset}v_j$.
We set the prices on the edges of $G$ as follows. Let $\eps'=\eps/5$.
For each path $Q_j\in \qset$, if it contains a single edge $e_i$, then we set the price of $e_i$ to be $(1-\eps')\cdot v_j$; if it contains two edges $e_i,e_{i'}$, then we set the price of both $e_i$ and $e_{i'}$ to be $(\frac{1}{2}-\eps')\cdot v_j$.
For each edges that does not belong to any path of $\qset$, we set its price to be $+\infty$.
	
We now show that the above prices will achieve worst-case welfare at least $(\frac{1}{2}-\eps')\cdot\optwel$. For a path $Q_j$ that contains a single edges $e_i$, clearly $e_i$ will be taken by some buyer (not necessarily $j$) at price $(1-\eps')\cdot v_j$, for any arriving order $\sigma$. For a path $Q_{j'}$ that contains two edges $e_{i},e_{i'}$, we notice that buyer $j'$ can afford her demand path. Thus for any order $\sigma$, at least one of $e_{i},e_{i'}$ will be sold at price $(\frac{1}{2}-\eps')\cdot v_{j'}$, otherwise buyer $j'$ must have purchased $Q_{j'}$. Hence for any order $\sigma$, the total price of the sold edges is at least $(\frac{1}{2}-\eps')\cdot\optwel$. Since all buyers have non-negative utility, the worst-case welfare is also at least $(\frac{1}{2}-\eps')\cdot\optwel\geq \frac{\optwel}{2+\eps}$.
\end{proof}

\begin{theorem}\label{thm:spider-no-tiebreak}
For any $\eps>0$, $\rationt(\spider)\leq 7/2+\eps$.
\end{theorem}

\begin{proof}[Proof of Theorem~\ref{thm:spider-no-tiebreak}]
Let $\qset$ be the set of paths in the optimal allocation. 
We define $\qset_1$ to be the set of all paths in $\qset$ that contains the center of the spider, and we define $\qset_2=\qset\setminus \qset_1$. For each path $Q\in \qset_1$, let $j_Q$ be the buyer who is allocated her demand path $Q$ in the optimal allocation. Define $G_Q$ to be the graph obtained by taking the union of all (one or two) legs whose edge sets intersect with $E(Q)$, and we denote $E_Q=E(G_Q)$. Since paths in $\qset_1$ are edge-disjoint, clearly for any $Q,Q'\in \qset_1$, $E_Q\cap E_{Q'}=\emptyset$. For any edge set $E$, let $\fset|_E\subseteq \fset$ be the sub-instance that contains all buyer $j$ where $Q_j$ has edges only in $E$. Then $\opt(G,\fset)=\sum_{Q\in\qset_1}\opt(G_Q,\fset|_{E_Q})$.

Let $e_{Q,1}$ and $e_{Q,2}$ be the two edges that are in Q and has the spider center as one endpoint.\footnote{$e_{Q,1}=e_{Q,2}$ if the spider center is one endpoint of path $Q$.}
Define $G'_Q$ to be the graph with edge set $E'_Q=\{e|e\in E_Q\setminus Q\}$. In other words, $G'_Q$ contains all edges not in $Q$ but in $E_Q$. Clearly we have $\opt(G_Q,\fset|_{E_Q})=\opt(G'_Q,\fset|_{E'_Q})+v_{j_Q}$. Define $G''_Q$ to be the graph with edge set $E''_Q=Q\setminus\{e_{Q,1},e_{Q,2}\}$; in other words, $G''_Q$ is the graph formed edges in $Q$, but excluding the center of the spider. Note that $E_Q=E'_Q\cup E''_Q\cup \{e_{Q,1},e_{Q,2}\}$. See \Cref{fig:spider_proof} for an illustration.

\begin{figure}[h]
	\centering
	\includegraphics[scale=0.12]{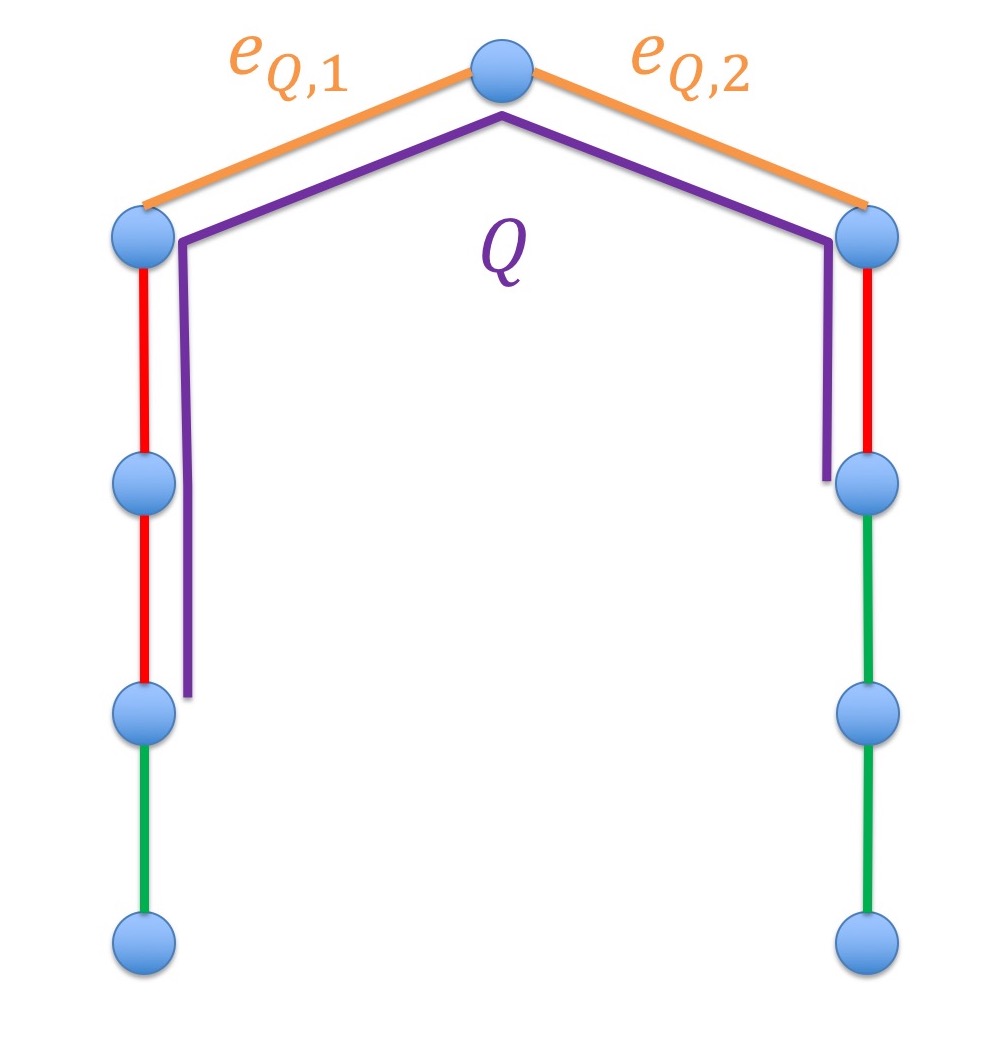}
	\caption{An illustration of edge sets. Edges of $E'_Q$ are shown in green, and edges of $E''_Q$ are shown in red.\label{fig:spider_proof}}
\end{figure}
Let $\alpha_Q=\frac{\opt(G''_Q,\fset|_{E''_Q})}{v_{j_Q}}$. Clearly $\alpha_Q\leq 1$, since $E''_Q\subseteq Q$. We will construct the price vector $p$ as follows. For every $Q\in \qset_1$, we will first construct two price vectors $p_1$ and $p_2$ on $E_Q$. Then depending on the instance, we will either choose $p(e)=p_1(e),\forall e\in E_Q$, or $p(e)=p_2(e),\forall e\in E_Q$. For every $e\not\in \cup_{Q\in \qset_1}E_Q$, we set $p(e)=+\infty$. Now fix any $Q\in \qset_1$ and any $\eps'>0$.  

\paragraph{Construction of $p_1$.} For each $e\in E'_Q$, let $p_1(e)=+\infty$. For each $e\in E''_Q$, let $p_1(e)=p^*(e)+\frac{1}{m}\eps'$, where $p^*=\{p^*(e)\}_{e\in E''_Q}$ is the set of prices from Theorem~\ref{thm:path} on $G''_Q$, whose worst-case welfare is $\opt(G''_Q,\fset|_{E''_Q})$. For $e=e_{Q,1}$ or $e_{Q,2}$, $p_1(e)=\frac{(1-\alpha_Q)v_{j_Q}}{2}-\eps'$. Then no buyer $j$ where $Q_j$ contains an edge in $E'_Q$ can afford her path. Also, we notice that since $p^*$ is the optimal solution in (LP-Dual) for instance $(G''_Q,\fset|_{E''_Q})$, every buyer $j$ where $Q_j$ has edges only in $E''_Q$ satisfies $v_j\leq p^*(Q_j)$. Thus none of them can afford her path in $p_1$, as $p_1(e)<p^*(e)$ for all $e\in E''_Q$.
	Moreover, the total price of edges in $Q$ is $\sum_{e\in E''_Q}p_1(e)+p_1(e_{Q,1})+p_1(e_{Q,2})<\opt(G''_Q,\fset|_{E''_Q})+\eps'+(1-\alpha_Q)v_{j_Q}-2\eps'<v_{j_Q}$.
	Consider any item pricing $p$ such that $p(e)=p_1(e),\forall e\in E_Q$. From the arguments above, only buyer $j$ whose $Q_j$ contains $e_{Q,1}$ or $e_{Q,2}$ may afford her path. Thus at least one of $e_{Q,1}$ and $e_{Q,2}$ must be sold under $p$ at price $\frac{(1-\alpha_Q)v_{j_Q}}{2}-\eps'$, for any buyers' arriving order $\sigma$. Otherwise, all edges in $E''_Q$ must also be unsold and buyer $j_Q$ should have purchased her path $Q$, contradiction. Thus the contributed welfare from edges in $E_Q$ is at least $\frac{(1-\alpha_Q)v_{j_Q}}{2}-\eps'$.

\paragraph{Construction of $p_2$.} For each $e\in E'_Q\cup E''_Q$, let $p_2(e)$ be the price of edge $e$ by applying Theorem~\ref{thm:path_no_tie} to $G'_Q\cup G''_Q$, which is a union of at most two path graphs. For $e=e_{Q,1}$ or $e_{Q,2}$, set $p_2(e)=+\infty$.
Then by Theorem~\ref{thm:path_no_tie},
if $p(e)=p_2(e),\forall e\in E_Q$,
the contributed welfare from edges in $E_Q$ is at least $\frac{2}{3}\opt(G'_Q\cup G''_Q,\fset|_{G'_Q\cup G''_Q})= \frac{2}{3}\opt(G'_Q,\fset|_{G'_Q})+\frac{2}{3}\opt(G''_Q,\fset|_{G''_Q})=\frac{2}{3}\opt(G'_Q,\fset|_{E'_Q})+\frac{2}{3}\alpha_Qv_{j_Q}$, for any buyers' arriving order $\sigma$. 

Now for any $Q\in \qset_1$, we choose the price with a higher contributed welfare: If $\frac{(1-\alpha_Q)v_{j_Q}}{2}-\eps'>\frac{2}{3}\opt(G'_Q,\fset)+\frac{2}{3}\alpha_Qv_{j_Q}$, we choose $p(e)=p_1(e),\forall e\in E_Q$ and choose $p(e)=p_2(e)$ otherwise. Thus the worst-case welfare of item pricing $p$ is at least
\begin{eqnarray*}
& &\sum_{Q\in \qset_1}\max\left(\frac{(1-\alpha)v_{j_Q}}{2}-\eps',\frac{2}{3}\opt(G'_Q,\fset|_{E'_Q})+\frac{2}{3}\alpha_Qv_{j_Q}\right)\\
&\geq&\sum_{Q\in \qset_1}\frac{4}{7}\left(\frac{(1-\alpha)v_{j_Q}}{2}-\eps'\right)+\frac{3}{7}\left(\frac{2}{3}\opt(G'_Q,\fset|_{E'_Q})+\frac{2}{3}\alpha_Qv_{j_Q}\right)\\
&=&\sum_{Q\in \qset_1}\left(\frac{2}{7}v_{j_Q}+\frac{2}{7}\opt(G'_Q,\fset|_{E'_Q})-\frac{4}{7}\eps'\right)\\
&=&\sum_{Q\in \qset_1}\left(\frac{2}{7}\opt(G_Q,\fset|_{E_Q})-\frac{4}{7}\eps'\right)>\frac{2}{7}\opt(G,\fset)-\frac{4m}{7}\eps'.
\end{eqnarray*}
Choosing $\eps'=\frac{\eps\cdot \opt(G,\fset)}{10m}$ finishes the proof. 
\end{proof}

\section{Missing Details from Section~\ref{sec: general graph}}
\label{apd: missing sec: general graph}

\subsection{Proof of Theorem~\ref{thm:bounded-degree-mahua}}
\label{apd: proof of thm:bounded-degree-mahua}

Denote by $K_i$ the resulting graph after constructing the gadget, as shown in Figure~\ref{fig:gadget_after}. We use the following lemma. 
\begin{lemma}\label{lem:disjointsort}
Let $\sigma$ be any permutation on $[a]$, then graph $K_i$ contains a set $\pset_i=\set{P_{i,j}\mid j\in [a]}$ of edge-disjoint paths, such that path $P_{i,j}$ connects vertex $u_{i,j}$ to vertex $u_{i+1,\sigma(j)}$.
\end{lemma} 

Before proving \Cref{lem:disjointsort}, we first give the proof of \Cref{thm:bounded-degree-mahua} using \Cref{lem:disjointsort}. 

\begin{prevproof}{Theorem}{thm:bounded-degree-mahua}
First, the graph $R_{a,b}$ is obtained by taking the union of all graphs $\set{K_i}_{0\le i\le b}$, while identifying, for each $i\in [b]$ and $j\in [a]$, the vertex $u_{i,j}$ in $K_{i-1}$ with the vertex $u_{i,j}$ in $K_{i}$.
Clearly, the maximum vertex degree in graph $R_{a,b}$ is $4$. 

We now show that we can easily convert the buyer-profile $\fset$ on graph $H_{a,b}$ into a buyer-profile $\hat\fset$ on graph $R_{a,b}$, while preserving all desired properties.
Consider first the buyers $1,\ldots,a$ in $\fset^*$. 
Let $\sigma_i$ be the identity permutation on $[a]$, for each $i\in [b]$. From \Cref{lem:disjointsort}, there exist sets $\set{\pset^*_i}_{i\in [b]}$ of edge-disjoint paths, where $\pset^*_i=\set{P^*_{i,j}\mid j\in [a]}$ for each $i\in [b]$.
We then let $\hat{\fset^*}$ contains, for each $j\in [a]$, a buyer $\hat B^*_j$ demanding the path $\hat Q_j$ with value $1$, where $\hat Q_j$ is the sequential concatenation of paths $P^*_{1,j},\ldots,P^*_{b,j}$. Clearly, paths $\set{\hat Q_j}_{j\in [a]}$ are edge-disjoint.
Consider now a set $S\subseteq [a]$ with $|S|\ge \sqrt{a}$.
Recall that in $\fset_S$ we have $2|S|$ buyers, whose demand paths cover the paths $\set{Q^{(j)}\mid j\in S}$ exactly twice.
Therefore, for each $i\in [b]$, the way that these paths connect vertices of $U_{i-1}=\set{u_{i-1,j}\mid j\in [a]}$ to vertices of $U_{i}=\set{u_{i,j}\mid j\in [a]}$ form two perfect matchings between vertices of $U_{i-1}$ and vertices of $U_{i}$.
From \Cref{lem:disjointsort}, there are two sets $\pset_i,\pset'_i$ of edge-disjoint paths connecting vertices of $U_{i-1}$ and vertices of $U_{i}$. We then define, for each demand path $Q=(j_0, j_1,\ldots,j_b)$\footnote{We can arbitrarily assign additionally the path $Q$ with some index $j_0\in S$, such that each index of $S$ is assigned to exactly two demand paths in $\dset_S$.} in $\fset_S$, its corresponding path $\hat Q$ to be the sequential concatenation of, the corresponding path in $\pset_1\cup\pset'_1$ that connects $u_{0,j_0}$ to $u_{1,j_1}$, the corresponding path in $\pset_2\cup\pset'_2$ that connects $u_{1,j_1}$ to $u_{2,j_2}$, all the way to the corresponding path in $\pset_b\cup\pset'_b$ that connects $u_{b-1,j_b-1}$ to $u_{b,j_b}$. It is easy to verify that all desired properties are still satisfied.
Lastly, to ensure that the graph $R_{a,b}$ can be embedded into the $(\sqrt{m}\times\sqrt{m})$-grid, we need $\sqrt{m}=a^2b$, where $b=a+3a^3$. Thus $a=\Theta(m^{1/10})$.
\Cref{thm:bounded-degree-mahua} now follows from Lemma~\ref{lem:ratio-Lab}. 
\end{prevproof}

It remains to prove \Cref{lem:disjointsort}.

\begin{proof}[Proof of Lemma~\ref{lem:disjointsort}]
We prove by induction on $a$. The base case where $a=1$ is trivial. Assume that the lemma is true for all integers $a\le r-1$. Consider the case where $a=r$.
For brevity of notations, we rename $K_i$ by $K$, $\pset_i$ by $\pset$, $u_{i,j}$ by $u_{j}$ and $u_{i+1,j}$ by $u'_{j}$.
Recall that the graph $K$ is the union of $r$ paths $W_1,\ldots,W_{r}$, where path $W_j$ connects vertex $u_{j}$ to vertex $u'_{r+1-j}$, such that every pair of these paths intersect at a distinct vertex.
Recall that we are also given a permutation $\sigma$ on $[r]$, and we are required to find a set $\pset=\set{P_j\mid j\in [r]}$ of edge-disjoint paths in $K$, such that the path $P_j$ connects $u_j$ to $u'_{\sigma(j)}$.

We first define the graph $K'=K\setminus W_{r}$, and we define an one-to-one mapping $f:[r\!-\!1]\to \set{2,\ldots,r}$ as follows. For each $j\in [r\!-\!1]$ such that $\sigma(j)\in \set{2,\ldots,r}$, we set $f(j)=\sigma(j)$; for $j\in [r\!-\!1]$ such that $\sigma(j)=1$, we set $f(j)=\sigma(r)$.
Note that $K'$ is a graph consisting of $(r\!-\!1)$ pairwise intersecting paths, and $f$ is a one-to-one mapping from the left set of vertices of $K'$ to the right set of vertices of $K'$.
From the induction hypothesis, there is a set $\pset'=\set{P'_j\mid j\in [r\!-\!1]}$ of edge-disjoint paths, such that path $P'_j$ connects $u_j$ to $u'_{f(j)}$.
If $\sigma(r)=1$, then we simply let $\pset=\pset'\cup \set{W_r}$, and it is easy to check that the set $\pset$ of paths satisfy the desired properties.
Assume now that $\sigma(r)\ne 1$, then the path $P'_{\sigma^{-1}(1)}$ is currently connecting $u_{\sigma^{-1}(1)}$ to $u'_{\sigma(r)}$, which is a wrong destination.
Observe that $W_r$ connects $u_r$ to $u'_1$ and is edge-disjoint with $P'_{\sigma^{-1}(1)}$. Moreover, it is easy to see that $P'_{\sigma^{-1}(a)}$ must intersect with $W_r$ at at least one vertex. Let $x$ be the intersection that is closest to $u_r$ on $W_r$. We then define the path $P_r$ as the concatenation of (i) the subpath of $W_k$ between $u_r$ and $x$, and (ii) the subpath of $P'_{\sigma^{-1}(1)}$ between $x$ and $u'_{\sigma(r)}$. Similarly, we define the path $P_{\sigma^{-1}(1)}$ as the concatenation of (i) the subpath of $P'_{\sigma^{-1}(1)}$ between $u_{\sigma^{-1}(1)}$ and $x$, and (ii) the subpath of $W_r$ between $x$ and $u'_1$. Then $P_r$ routes $u_r$ to $u'_\sigma(r)$, while $P_{\sigma^{-1}(1)}$ routes $u_{\sigma^{-1}(1)}$ to $u'_1$ correctly.
We then let $\pset=(\pset'\setminus \set{P'_{\sigma^{-1}(1)}})\cup \set{P_r,P_{\sigma^{-1}(1)}}$, and it is easy to check that the set $\pset$ of paths satisfy the desired properties.
\end{proof}

\subsection{Proof of \Cref{thm:general_graph}}
\label{apd: proof of thm:general_graph}

\begin{proof}

Choose parameter $\alpha=1/10$. Denote $\fset=\set{(Q_j,v_j)}_{j\in [n]}$.
Let $\qset=\set{Q_{j_1},\ldots,Q_{j_t}}$ be the independent set of paths that, among all independent subsets of $\set{Q_1,\ldots,Q_n}$, maximizes its total value, namely $\qset=\arg_{\tilde \qset : \tilde \qset \text{ is independent}}\max\set{\sum_{j:Q_j\in \tilde\qset}v_j}$. For every $Q\in \qset$, we denote $v(Q)$ the value of the buyer who is allocated her demand path $Q$ in the optimal allocation $\qset$. We denote $J=\set{j_1,\ldots,j_t}$, so $\optwel(G,\fset)=\sum_{j\in J}v_j$. 
We first pre-process the set $\qset$ as follows.

For each integer $1\le k\le \log m$, we denote by $\qset_{k}$ the subset of paths in $\qset$ that whose length lies in the interval $[2^{k-1}, 2^k)$. Clearly, there exists some integer $k^*$ with $v(\qset_{k^*})=\sum_{Q\in \qset_{k^*}}v(Q)\ge \Omega(\optwel/\log m)$. We denote $\qset^*=\qset_{k^*}$ and denote $L=2^{k^*-1}$, so the length of each path in $\qset^*$ lies in $[L,2L)$.

Let $v^*$ be the maximum value of a path in $\qset^*$, namely $v^*=\max\set{v_j\mid Q_j\in \qset^*}$. We let $\qset^*_0$ contains all paths in $\qset^*$ with value at most $v^*/m^2$. Then, for each integer $1\le t\le 2\log n$, we denote by $\qset^*_{t}$ the subset of paths in $\qset^*$ whose length lie in the interval $(v^*/2^{t}, v^*/2^{t-1}]$.
Clearly, the total value of the paths in $\qset^*_0$ is at most $m\cdot (v^*/m^2)= v^*/m$ (since $|\qset^*_0|\le |\qset^*|\le m$). Therefore, there exists some integer $t^*$ with $v(\qset^*_{t^*})=\sum_{Q\in \qset^*_{t^*}}v(Q)\ge \Omega(v(\qset^*)/\log m)$. We denote $\qset'=\qset^*_{t^*}$ and denote $\hat v=v^*/2^{t^*}$, so the value of each path of $\qset'$ lies in $(\hat v,2\hat v]$. 

So far we obtain a set $\qset'$ of paths, and two parameters $L,\hat v$, such that 
\begin{enumerate}
\item all paths in $\qset'$ have length in $[L, 2L)$; 
\item all paths in $\qset'$ have value in $(\hat v, 2\hat v]$;
\item the total value of all paths in $\qset'$ is $\Omega(\optwel/\log^2 m)$.
\end{enumerate}

We use the following observations.
\begin{observation}
\label{obs:few_paths}
If $|\qset'|<m^{1/2-\alpha}$, then there exist prices on edges of $G$ achieving worst-case welfare $\Omega(\optwel/(m^{1/2-\alpha}\cdot\log^2m))$.
\end{observation}
\begin{proof}
Let $Q$ be a path in $\qset'$ with largest value, so $v(Q)\ge v(\qset')/m^{1/2-\alpha}$. 
We first set the price of each edge of $E(G\setminus Q)$ to be $+\infty$. 
From Theorem~\ref{sec:pathnt}, we know that there is a set of prices on edges of $Q$, that achieves the worst-case welfare $\Omega(v(Q))$. 
Therefore, we obtain a set of prices for all edges in the graph, that achieves worst-case welfare $\Omega(v(Q))=\Omega(v(\qset')/m^{1/2-\alpha})=\Omega(\optwel/(m^{1/2-\alpha}\cdot\log^2m))$.
\end{proof}

\begin{observation}
\label{obs:small_lengths}
If $L < m^{1/2-\alpha}$, then there exist prices on edges of $G$ achieving worst-case welfare $\Omega(\optwel/(m^{1/2-\alpha}\cdot\log^2m))$.
\end{observation}
\begin{proof}
We define the edge prices as the following.
For each path $Q\in \qset'$ and for each edge $e\in Q$, we set its price $\tilde p(e)$ to be $v(Q)/2L$, and we set the price for all other edges to be $+\infty$.
Note that, for each $Q\in \qset'$, $\tilde p(Q)=\sum_{e\in Q}\tilde p(e) =|E(Q)|\cdot(v(Q)/2L)< v(Q)$.
Therefore, no matter what order in which the buyers come, at the end of the selling process, for each $Q\in \qset'$, at least one edge is taken at the price of $v(Q)/2L$. It follows that the welfare is at least 
$\sum_{Q\in \qset'} v(Q)/(2L)=\Omega(v(\qset')/m^{1/2-\alpha})=\Omega(\optwel/(m^{1/2-\alpha}\cdot\log^2m))$.
\end{proof}

From the above observations, we only need to consider the case where $|\qset'|\ge m^{1/2-\alpha}$ and $L\ge m^{1/2-\alpha}$. Since $|\qset'|\cdot L\le \sum_{Q\in \qset'}|E(Q)|\le m$, we get that $|\qset'|, L\le m^{1/2+\alpha}$. We can also assume without loss of generality that $\hat v=1$, namely each path of $\qset'$ has value in $(1,2]$. We now perform the following steps.

\paragraph{Step 1. Construct a random subset $\qset''$ of $\qset'$.} 
Let $\qset''$ be a subset of $\qset'$ obtained by including each path $Q\in \qset'$ independently with probability $1/2$. Then we set the price for each edge that is not contained in any path of $\qset''$ to be $+\infty$.
From Chernoff's bound, with high probability $v(\qset'')\ge v(\qset')/3$, so $v(\qset'')\ge \Omega(\optwel/\log^2 m)$. Let $\tilde\qset$ be the set of all paths in $\set{Q_1,\ldots,Q_n}$ that do not contain an edge whose price is $+\infty$, namely the set $\tilde\qset$ contains all paths that survive the current price. Here we say a path $Q$ survives a price $p$, if $\sum_{e\in E(Q)}p(e)<v(Q)$, i.e. the buyer's value is higher than the price of her demand path.
Note that the set $\tilde\qset$ may contain paths of any length.
We use the following observation.
\begin{observation}
\label{obs:touching_too_many_opt_paths}
With high probability, all paths in $\tilde\qset$ intersect at most $2\log n$ paths in $\qset''$.
\end{observation}
\begin{proof}
For each path $Q$, if it intersects with at least $2\log n$ paths in $\qset'$, then the probability that it belongs to $\tilde\qset$ is at most $(1/2)^{2\log n}=1/n^2$. From the union bound, the probability that there exists a path of $\set{Q_1,\ldots,Q_n}$ that intersects at least $2\log n$ paths in $\qset'$ and is still contained in $\tilde\qset$ is at most $n\cdot(1/n^2)=1/n$. Observation~\ref{obs:touching_too_many_opt_paths} then follows.
\end{proof}

\paragraph{Step 2. Analyze a special set of short surviving paths.} 
We set $L'=m^{1/4-\alpha/2}$, and let $\tilde\qset'$ contains all paths in $\tilde\qset$ with length less than $L'$. Clearly, $\tilde\qset'\cap\qset''=\emptyset$.
Let $\hat\qset$ be a max-total-value independent subset of $\tilde\qset'$.
We distinguish between the following two cases on the value of $v(\hat\qset)$.

\textbf{Case 1. $v(\hat\qset)\ge v(\qset'')/100L'$.}
We will show that in this case, there exist prices on edges of $G$ that achieve worst-case welfare $\Omega(\optwel/(m^{1/2-\alpha}\cdot\log^2m))$. 
We define the edge prices as follows.
For each path $Q\in \hat\qset$ and for each edge of $Q$, we set its price to be $v(Q)/L'$, and we set the price for all other edges to be $+\infty$.
Similar to Observation~\ref{obs:small_lengths}, the worst-case welfare is at least 
$\sum_{Q\in \hat\qset} v(Q)/L'=\Omega(v(\hat\qset)/L')=
\Omega(v(\qset'')/L'^2)=
\Omega(\optwel/(m^{1/2-\alpha}\cdot\log^2m))$.

\textbf{Case 2. $v(\hat\qset)\le v(\qset'')/100L'$.}
We set the prices of the edges in two stages.
In the first stage, for each path $Q\in\hat\qset$ and for each edge in $Q$, we set its price $p(e)$ to be $v(Q)$, and for all other edges that belong to some path of $\qset''$, we set its price to be $0$.
Recall that we have already set the price for all edges that do not belong to a path of $\qset''$ to be $+\infty$.
We prove the following observations.
\begin{observation}
No path in $\tilde \qset'$ survives price $p$.
\end{observation}
\begin{proof}
Assume by contradiction that there is a path $Q$ with length $|E(Q)|< L'$ that survives the price $p$. Let $\hat\qset'$ be the set of paths in $\hat\qset$ that intersect $Q$. From the definition, $v(Q)>\sum_{e\in E(Q)}p(e)\ge v(\hat\qset')$. Consider the set $(\hat\qset\setminus \hat\qset')\cup\set{Q}$. From the above discussion, this set is an independent set of paths with length less than $L'$, with total value at least the total value of $\hat\qset$, while containing strictly less paths than $\hat\qset$. This leads to a contradiction to the optimality of $\hat\qset$.
\end{proof}

\begin{observation}
\label{obs: surviving value large}
The total value of all paths in $\qset''$ that survives $p$ is at least $0.99\cdot v(\qset'')$.
\end{observation}
\begin{proof}
For each path $Q\in\qset''$ that does not survive the price $p$, its value is below $\sum_{e\in E(Q)}p(e)$. Since paths in $\qset''$ are edge-disjoint, the total value of paths in $\qset''$ that do not survive the price $p$ is at most $\sum_{e\in E(\qset'')}p(e)\le \sum_{Q\in \hat\qset}L'\cdot v(Q)\le L'\cdot v(\hat \qset)\le v(\qset'')/100$. Observation~\ref{obs: surviving value large} then follows.
\end{proof}

We now modify the price $p$ in the first stage as the following. For each path $Q\in \qset''$ and for each edge $e\in Q$, we increase its price by $v(Q)/4L$. Note that the prices of all other edges are already set to be $+\infty$ before the first stage. This completes the definition of the prices on edges. We now show that these prices will achieve worst-case welfare $\Omega(\optwel/(m^{1/2-\alpha}\cdot\log^2m\cdot\log n))$, thus completing the proof of \Cref{thm:general_graph}.

First, since the total price on all edges of $\qset''$ is at most $v(\qset'')/100+\sum_{Q\in \qset''}2L\cdot(v(Q)/4L)\le 0.51\cdot v(\qset'')$, the set of paths in $\qset''$ that survives the ultimate price has total value at least $0.49\cdot v(\qset'')$.
We denote this set by $\hat{\qset''}$.
Since all path in $\qset''$ has length $[L,2L)$ and value $(\hat v,2\hat v]$, $|\hat\qset''|\ge 0.49|\qset''|/4\ge 0.1|\qset''|$.
Denote by $\bar\qset$ the resulting set of paths that is being taken in the selling process. It is clear that 
\begin{itemize}
\item each path $Q'\in \bar\qset$ has length at least $L'$;
\item for each path $Q\in\hat\qset''$, there is a path $Q'\in \bar\qset$ that intersects $Q$; and
\item each path $Q'\in \bar\qset$ intersects at most $2\log n$ paths in $\hat\qset''$.
\end{itemize}
Altogether, we get that $|\bar{\qset}|\ge |\hat\qset''|/(2\log n)$, the length of each path in $\bar{\qset}$ is at least $L'$, and each edge of $E(\bar{\qset})$ has price at least $\min_{e\in E(\hat\qset'')}\set{p(e)}$. Therefore, the total value of paths in $\bar{\qset}$ is at least 
\[\frac{|\hat\qset''|}{2\log n}\cdot L'\cdot \min_{e\in E(\hat\qset'')}\set{p(e)}\ge \frac{|\hat\qset''|}{2\log n}\cdot L'\cdot \Omega\left(\frac{v(\qset'')}{|\qset''|\cdot L}\right)
=\Omega\left(\frac{v(\qset'')\cdot L'}{L\cdot \log n}\right)
\ge\Omega\left(\frac{\optwel}{m^{1/2-\alpha}\cdot\log^2m\log n}
\right),\]
where the last inequality is obtained by combining $L\le m^{1/2+\alpha}$, $L'=m^{1/4-\alpha/2}$ and $\alpha=1/10$.

\end{proof}

\section{Missing Details in Section~\ref{sec:congestion}}
\label{apd: missing sec:congestion}

\subsection{Proof of Theorem~\ref{thm: congestion_c}}

\begin{proof}
Let $\qset=\set{Q_{j_1},\ldots,Q_{j_t}}$ be an independent set with maximum total value.
Denote $J=\set{j_1,\ldots,j_t}$, so $\optwel(U,\fset)=\sum_{j\in J}v_j$. Denote $v=\optwel(U,\fset)$.

$U'=\bigcup_{Q\in \qset}Q$, so $U'\subseteq U$ and $|U'|\le m$.
Note that $v=\sum_{j:Q_j\in \qset}v_j$.
We now define the prices on items of $U$ as follows.
For each element $i\notin U'$, we define $p^1(i)=\ldots=p^c(i)=+\infty$.
For each element $i\in U'$ and for each $1\le k\le c$, we define $p^k(i)=\frac{1}{2m}\cdot m^{(k-1)/c}$.

We now show that these prices achieve worst-case welfare $\Omega(v/m^{1/c})$.
Consider the remaining items of $U$ at the end of the selling process.
Clearly, all copies of items in $U\setminus U'$ must be unsold. Fix any buyers' order. We distinguish between the following cases.

\textbf{Case 1. Some element of $U'$ is sold out.} From the definition of prices in $\set{p^{k}(i)\mid i\in U, 1\le k\le c}$, for each $i\in U'$, $p^c(i)=\frac{v}{2}\cdot m^{-1/c}$. Therefore, if all $c$ copies of the item $i$ are sold, then the $c$-th copy is sold at price $p^c(i)=\frac{v}{2}\cdot m^{-1/c}$. It follows that the welfare is $\frac{v}{2}\cdot m^{-1/c}=\Omega(v\cdot m^{-1/c})$.

\textbf{Case 2. No elements of $U'$ are sold out.}
For each $i\in U'$, we define $\tilde p(i)=p^{k+1}(i)$ iff $k$ copies of $i$ are sold. Since no elements of $U'$ are sold out, $\tilde p(i)$ is well-defined for all $i\in U'$. If $k\geq 1$ copies of item $i$ are sold, the total prices of the sold copies of $i$ is at least $p^{k}(i)=\tilde p(i)\cdot m^{-1/c}>(\tilde p(i)-\frac{v}{2m})m^{-1/c}$; If $k=0$ copies of item $i$ are sold,
it also holds that $(\tilde p(i)-\frac{v}{2m})m^{-1/c}=0$. Since each buyer has a non-negative utility, the total welfare is at least $\sum_{i\in U'}(\tilde p(i)-\frac{v}{2m})m^{-1/c}\geq m^{-1/c}\sum_{i\in U'}\tilde p(i)-\frac{v}{2}m^{-1/c}$.

Let $\tilde\qset\subseteq\qset$ be the demand sets of $\qset$ that are sold. If $\sum_{j:Q_j\in \tilde\qset}v_j\ge \frac{v}{4}$, then the welfare of the pricing is already $\Omega(v\cdot m^{-1/c})$. Otherwise, $\sum_{j:Q_j\in \tilde\qset}v_j< \frac{v}{4}$, then $\sum_{j:Q_j\in \qset\setminus\tilde\qset}v_j>\frac{3v}{4}$. Since none of the buyers in $\qset\setminus\tilde\qset$ are sold, it means that the total prices of all items at the end of the selling process is at least $\frac{3v}{4}$, since $\sum_{i\in U'}\tilde p(i)
\ge
\sum_{Q\in \qset\setminus \tilde\qset}\tilde p(Q)
\ge
\sum_{j:Q_j\in \qset\setminus \tilde\qset}\tilde v_j
>\frac{3v}{4}$. The theorem follows since the total welfare is at least $m^{-1/c}\sum_{i\in U'}\tilde p(i)-\frac{v}{2}\cdot m^{-1/c}=\Omega(v\cdot m^{-1/c})$.
\end{proof}

\subsection{Proof of Theorem~\ref{thm: congestion_c lower bound}}

\begin{proof}
Let $r$ be the number such that $m=\binom{r}{c+1}$, so $r=O(c\cdot m^{1/(c+1)})$.
Denote $I=\set{1,\ldots,r}$ and let $U=\set{u_J\mid J\subseteq I, |J|=c+1}$, namely $U$ contains $\binom{r}{c+1}=m$ elements, where each element is indexed by a size-$(c+1)$ subset $J$ of $I$.
We define the buyer profile as follows.
There are $(c+1)r+1$ buyers: buyer $B_0$, and, for each $1\le k\le c+1, 1\le j\le r$, a buyer named $B^k_j$.
The demand set for buyer $B_0$ is $S_0=U$, and her value is $|U|=\binom{r}{c+1}$.
For each $1\le j\le r$ and for each $1\le k\le c+1$, the demand set of buyer $B^k_j$ is $Q^k_j=\set{u_J\mid j\in J}$, and her value is $v(S^k_j)=|S^k_j|=\binom{r-1}{c}$.
Clearly optimal allocation when each item has supply 1 is to assign all elements of $U$ to $B_0$, and the optimal welfare is $\opt(U,\fset)=|U|=\binom{r}{c+1}$.
Moreover, it is easy to verify that any $c+1$ demand sets of $\set{Q^k_j}_{1\le k\le c+1,1\le j\le r}$ shares at least one element of $U$.

We claim that any set of prices can achieve worst-case welfare at most $c\cdot\binom{r-1}{c}$.
Note that this finishes the proof of Theorem~\ref{thm: congestion_c lower bound}, since $$\frac{c\cdot\binom{r-1}{c}}{\opt(U,\fset)}=\frac{c\cdot\binom{r-1}{c}}{\binom{r}{c+1}}=\frac{c\cdot\binom{r-1}{c}}{\frac{r}{c+1}\cdot\binom{r-1}{c}}=\frac{c(c+1)}{r}=O\left(\frac{c+1}{m^{1/(c+1)}}\right).$$
It remains to prove the claim.

Let $\set{p^k(u)\mid u\in U, 1\le k\le c}$ be any set of prices.
We will iteratively construct an order $\sigma$ on buyers, such that if the buyers come to the auction according to this order, the achieved welfare is at most $c\cdot\binom{r-1}{c}$.
Initially. $\sigma$ is an empty sequence.
Throughout, we maintain a set $\set{\tilde p(u)}_{u\in U}$ of prices, such that at any time, $\tilde p(u)$ is the price of the cheapest available copy of item $u$. Initially, $\tilde p(u)=p^1(u)$ for all $u\in U$.

We perform a total of $c$ iterations, and now we fix some $1\le k\le c$ and describe the $k$-th iteration. We first check whether or not there is a set $Q_j^k$ with $\tilde p(Q_j^k)\le v(Q_j^k)$. If so, assume $\tilde p(Q_{j_k}^k)\le v(Q^k_{j_k})$, then we add the buyer $B^k_{j_k}$ to the end of the current sequence $\sigma$, update the price $\tilde p(u)$ for all items $u\in Q^k_{j_k}$ to their next price in $\set{p^k(u)\mid 1\le k\le c}$. And then continue to the next iteration. Otherwise, since every element of $U$ appears in exactly $c+1$ sets of $Q^k_1,\ldots,Q^k_{r}$, and $\sum_{1\le j\le r}v(Q^k_j)=r\binom{r-1}{c}=(c+1)\binom{r}{c+1}=(c+1)\cdot |U|$, we get that $\tilde p(U)> |U|$.
Therefore, no buyer can afford her demand set, and the welfare will be $0$. In this case, we add all buyers to the end of $\sigma$ and terminate the algorithm.

We now analyze the algorithm. If the algorithm is terminated before it completes $c$ iterations, then from the construction above, the buyer $B_0$ will not get her demand set. Also, since any $c+1$ other demand sets share an element of $U$, at most $c$ other buyers may get their demand sets.  Therefore, the welfare is at most $c\cdot\binom{r-1}{c}$.
Assume that the algorithm successfully completes $c$ iterations. From the description of the algorithm, in each iteration, some buyer from $\set{B^k_j\mid 1\le k\le c+1,1\le j\le r}$ will be added to the sequence, and moreover, this buyer will get her demand set under order $\sigma$.
Therefore, after $c$ iterations, we added $c$ distinct buyers to the sequence that will get their demand sets.
Since any $c+1$ demand sets shares an element of $U$, we know that no other buyer may get her demand set anymore, so the welfare is at most $c\cdot\binom{r-1}{c}$.
\end{proof}

\subsection{Proof of Theorem~\ref{thm: lower_bound_congestion}}

\begin{proof}
We will use the graph $H_{a,b}$ constructed in \Cref{thm:lb-mahua}. For convenience, we will work with the multi-graph $L_{a,b}$.
The parameters $a,b$ are set such that $b=a+2(c+1)^{c+2}a^{c+2}$ and $m=ab$, so $a=\Theta(m^{1/(c+3)})$.

We now define the buyer profile. Recall that in \Cref{thm:lb-mahua} the buyer profile $\fset=\fset^*\cup(\bigcup_{S\subseteq [a], |S|\ge \sqrt{a}}\fset_S)$, where $\fset^*$ contains, for each $r\in [a]$, a buyer $B^*_r$ demanding the path $Q^{(r)}=(r,r,\ldots,r)$ with value $1$.
The buyer profile $\hat\fset$ that we will use in this subsection is similar to $\fset$. 
Specifically, we will keep the buyers in $\fset^*$, but will also additionally construct, for each set $S\subseteq [a]$ such that $|S|\ge \sqrt{ca}$ (instead of $|S|\ge \sqrt{a}$), a set $\hat\fset_S$ of buyers, whose demand paths and values satisfy the following properties.

\begin{enumerate}
	\item \label{prop'1}For each $S$, set $\hat\fset_S$ contains $(c+1)|S|$ buyers, and every pair $Q,Q'$ of demand paths in $\hat\fset_S$ share some edge, and the value for each demand path is $1+\epsilon$.
	\item \label{prop'3}For each demand path $Q$ in $\hat\fset_S$, the index sequence $(j^Q_1,\ldots,j^Q_b)$ that  $Q$ corresponds to satisfies that (i) $j^Q_i\in S$ for each $i\in [b]$; and (ii) the set $\set{j^Q_1,\ldots,j^Q_b}$ contains all element of $S$.
	\item \label{prop'4}The union of all demand paths in $\hat\fset_S$ covers the graph $\bigcup_{r\in S}Q^{(r)}$ exactly $c+1$ times. In other words, for each $i\in [b]$, the multi-set $\set{j^Q_i\mid Q\in \hat\fset_S}$ contains each element of $S$ exactly twice.
	\item \label{prop'5}For any $c+1$ subsets $S_1,\ldots,S_{c+1}$ of $[a]$, such that $|S_t|\ge \sqrt{ca}$ for each $t\in [c\!+\!1]$ and $\bigcap_{t}S_t\ne \emptyset$, for any $c+1$ demand paths $Q_1,\ldots,Q_{c+1}$ such that $Q_t\in \hat\fset_{S_t}$ for each $t\in [c\!+\!1]$, $\bigcap_{t}E(Q_t)\ne \emptyset$.
\end{enumerate}

Suppose that we have successfully constructed the sets $\set{\hat\fset_S}_{S\subseteq [a], |S|\ge \sqrt{ca}}$ that satisfy the above properties.
We then let $\fset$ be the union of $\fset^*$ and, for each set $S\subseteq [a], |S|\ge \sqrt{ca}$, $c$ distinct copies of set $\hat\fset_S$. In other words, for each buyer in $\set{\hat\fset_S}_{S\subseteq [a], |S|\ge \sqrt{ca}}$, we duplicate $c$ buyers and add all of them into our buyer profile.
This completes the description of $\hat\fset$.
From the above properties, it is easy to see that $\opt(L_{a,b},\hat\fset)=a$, which is achieved by giving each buyer in $\fset^*$ her demand path.
We will prove that any prices on edges of $L_{a,b}$ may achieve worst-case welfare $O(\sqrt{ca})$. Since $a=\Theta(m^{1/(c+3)})$, $\sqrt{a}=\Theta(m^{1/(2c+6)})$, which completes the proof of \Cref{thm: lower_bound_congestion}.

Consider now any set $\set{p^k(e)\mid e\in E(L_{a,b}), k\in [c]}$ of prices on edges of $L_{a,b}$.
From Properties \ref{prop'3} and \ref{prop'5}, it is easy to see that at most $\sqrt{ca}$ buyers from $\hat{\fset}\setminus \fset^*$ can get their demand paths simultaneously. We consider the following arrival order of the buyers in $\hat\fset$. The buyers in $\hat\fset$ are divided into $(c+1)$ groups, where each of the first $c$ groups contains a copy of each buyer from $(\bigcup_{S\subseteq [a], |S|\ge \sqrt{a}}\fset_S)$, and the last group contains all buyers of $\fset^*$. The buyers come to the auction according to their group index: All buyers from the first group come first (buyers within the same group come at an arbitrary order), and then all buyers from the second group come, etc.

Assume that we pause the selling process right after all buyers from the first $ca$ groups have come.
Currently for each edge $e\in E(L_{a,b})$, some copies of it were taken and there is a price $p^*(e)$ on its next copy.
We distinguish between the following two cases.

\paragraph{Case 1. At least $\sqrt{ca}$ buyers in $\fset^*$ can afford their demand paths at prices $\set{p^*(e)}_{e\in E(L_{a,b})}$.}
We let $S$ be the set that contains all indices $r\in [a]$ such that the buyer $B^*_r$ can afford her demand path $Q^{(r)}$ at prices $\set{p^*(e)}_{e\in E(L_{a,b})}$, so $|S|\ge \sqrt{ca}$. Similar to \Cref{lem:ratio-Lab}, it is easy to show that at least one buyer in $\hat\fset_S$ can afford her demand path, and therefore all $c$ copies of this buyer in $\hat\fset$ will get their demand paths. This implies that $p^*(Q^{(r)})=+\infty$ for all $r\in S$, contradicting with the assumption in this case.

\paragraph{Case 2. At most $\sqrt{ca}$ buyers in $\fset^*$ can afford their demand paths at prices $\set{p^*(e)}_{e\in E(L_{a,b})}$.}

Similar to Lemma~\ref{lem:ratio-Lab}, the optimal welfare is at most $(2+\epsilon)\sqrt{ca}$.

$\ $

It remains to construct the sets $\set{\hat\fset_S}_{S\subseteq [a], |S|\ge \sqrt{ca}}$ that satisfy the required properties. The construction is almost identical to that of \Cref{thm:lb-mahua}. The only difference is that we need to construct a $(c+1)s\times b$ matrix $M'_S$, instead of a $2s\times b$ matrix. Accordingly, we first place $(c+1)$ copies of matrix $N_S$ vertically as the first $s$ columns of $M'_S$, and then for the next $b-s$ columns, we let each column to be each column be an independent random permutation on elements of the multiset that contains, for each element of $S$, $(c+1)$ copies of it. It is easy to verify that, with $b=a+2(c+1)^{c+2}a^{c+2}$, all the desired properties are satisfied with high probability.

\end{proof}

\end{document}